\newif\iflong\longfalse
\longtrue

\RequirePackage{mathtools} 
\documentclass{llncs}

\usepackage{amssymb}
\usepackage{pifont,relsize}
\usepackage{array,enumitem}
\usepackage{xcolor}
\usepackage[skipabove=0pt,skipbelow=0pt]{mdframed}
\usepackage{tikz}
\usetikzlibrary{matrix,chains,positioning,decorations.pathreplacing,arrows}
\usetikzlibrary{tikzmark}
\usepackage{subcaption}
\usepackage{hyperref}
\usepackage{scrextend}
\usepackage[capitalise]{cleveref}
\hypersetup{colorlinks,linkcolor={red!50!black},citecolor={blue!50!black}}
\iflong
\usepackage{arydshln}
\usepackage{booktabs,multirow}
\fi

\makeatletter
\@ifundefined{smallmarginstrue}{\newif\ifsmallmargins\smallmarginsfalse}{}
\makeatother

\iflong
	\ifsmallmargins
		\usepackage[paperwidth=17.51cm,paperheight=24.76cm,marginpar=.8in]{geometry}
	\else
		\usepackage[marginpar=.8in]{geometry}
	\fi 
\else 
	\ifsmallmargins
		\usepackage[paperwidth=13.2cm,paperheight=21.3cm]{geometry}
	\fi 
\fi 

\iflong
\title{P\texorpdfstring{{\smaller INFER}}{INFER}: Privacy-Preserving
  Inference for Machine Learning\thanks{This is an extended version
    of~\cite{JP19}.}} 
\else
\title{P\texorpdfstring{{\smaller INFER}}{INFER}: Privacy-Preserving
  Inference}
\fi 
\titlerunning{\textsc{Pinfer}: Privacy-Preserving Inference}
\subtitle{Logistic Regression, Support Vector Machines, and More,
  Over Encrypted~Data}
\author{Marc Joye \and Fabien A. P. Petitcolas}
\institute{OneSpan, Brussels, Belgium\\
  \email{\{marc.joye,fabien.petitcolas\}{@}onespan.com}}
      
\renewcommand*{\bbbr}{\mathbb{R}}
\renewcommand*{\bbbz}{\mathbb{Z}}
\DeclarePairedDelimiter{\abs}{\lvert}{\rvert}
\DeclarePairedDelimiter{\pred}{\textbf{[}}{\textbf{]}}
\DeclareMathOperator{\lcm}{lcm}
\DeclareMathOperator{\sign}{sign}
\DeclareMathOperator{\ReLU}{ReLU}
\newcommand*{\getsr}{\stackrel{\scriptscriptstyle\mkern2mu R}{\gets}}
\newcommand*{\simto}{\stackrel{\smash{\raisebox{-.5ex}{%
  $\scriptscriptstyle\sim\mkern4mu$}}}{\to}}
\makeatletter
\newcommand*{\cding}[1]{\@tempcnta=#1\advance\@tempcnta by181%
  \textcolor{gray!25!black}{\ding{\@tempcnta}}}
\makeatother
\newcommand*{\tran}{{\raisebox{2pt}{$\scriptscriptstyle\intercal$}}}
\newcommand{\larrow}[1][]{\stackrel{\mathclap{\substack{\textstyle#1}}}{%
    \hbox to.225\textwidth{\leftarrowfill}}}
\newcommand{\rarrow}[1][]{\stackrel{\mathclap{\substack{\textstyle#1}}}{%
    \hbox to.225\textwidth{\rightarrowfill}}}
\newcommand{\ul}[1]{\underline{\smash{#1}}}
\newcommand*{\B}{\mathcal{B}}

\newcommand*{\D}{\mathcal{D}}

\newcommand*{\M}{\mathcal{M}}

\newcommand*{\X}{\mathcal{X}}
\newcommand*{\Y}{\mathcal{Y}}
\DeclareMathOperator{\sigmoid}{\sigma}
\makeatletter
\DeclareFontEncoding{LS2}{}{\@noaccents}
\makeatother
\DeclareFontSubstitution{LS2}{stix}{m}{n}
\DeclareSymbolFont{largesymbolsstix}{LS2}{stixex}{m}{n}
\DeclareMathDelimiter{\lBrack}{\mathopen}{largesymbolsstix}{"E0}{largesymbolsstix}{"06}
\DeclareMathDelimiter{\rBrack}{\mathopen}{largesymbolsstix}{"E1}{largesymbolsstix}{"07}
\DeclareMathDelimiter{\lBrace}{\mathopen}{largesymbolsstix}{"E8}{largesymbolsstix}{"0E}
\DeclareMathDelimiter{\rBrace}{\mathopen}{largesymbolsstix}{"E9}{largesymbolsstix}{"0F}
\DeclarePairedDelimiter{\dbracket}{\lBrack}{\rBrack}
\DeclarePairedDelimiter{\idbracket}{\rBrack}{\lBrack}
\DeclarePairedDelimiter{\dbrace}{\lBrace}{\rBrace_{\mathstrut\!s}}
\DeclarePairedDelimiter{\idbrace}{\rBrace}{\lBrace_{\mathstrut\!s}}
\newcommand*{\HEnc}{\dbracket}
\newcommand*{\HDec}{\idbracket}
\newcommand*{\HHEnc}{\dbrace}
\newcommand*{\HHDec}{\idbrace}
\newcommand*{\secpar}{\kappa}
\newcommand*{\pk}{\mathit{pk}}
\newcommand*{\sk}{\mathit{sk}}
\newcommand*{\smashscript}[2]{{\smash{#1}_{\scriptscriptstyle\mkern-2mu #2}}}
\newcommand*{\pkC}{\smashscript{\pk}{C}}
\newcommand*{\skC}{\smashscript{\sk}{C}}
\newcommand*{\pkS}{\smashscript{\pk}{S}}
\newcommand*{\skS}{\smashscript{\sk}{S}}
\newcommand*{\deltaC}{\smashscript{\delta}{C}}
\newcommand*{\deltaS}{\smashscript{\delta}{S}}
\newcommand*{\ct}{\mathfrak}
\newcommand*{\yhat}{\hat{y}}
\newcommand*{\tstar}{t^*}
\newcommand*{\thetax}{\vec{\theta}^\tran\vec{x}}
\makeatletter
\newcommand*{\cboxplus@}{\vcenter{\hbox{$\m@th\mathchar"0401$}}}
\newcommand*{\cboxminus@}{\vcenter{\hbox{$\m@th\mathchar"040C$}}}
\renewcommand*{\boxplus}{\mathbin{\cboxplus@}}
\renewcommand*{\boxminus}{\mathbin{\cboxminus@}}
\makeatother
\makeatletter
\newcommand*{\b@xsum@}[1]{\vphantom{\sum}%
  \vcenter{\hbox{\resizebox{!}{#1}{$\m@th\cboxplus@$}}}}
\newcommand*{\boxsum@}{\mathop{%
    \mathchoice{\b@xsum@{2.75ex}}{\b@xsum@{2ex}}%
    {\b@xsum@{1.35ex}}{\b@xsum@{1ex}}}}
\newcommand*{\boxsum}{\DOTSB\boxsum@\slimits@}
\makeatother
  
\definecolor{lightsalmon}{rgb}{1,.628,.48}
\definecolor{lemonchiffon}{rgb}{1,1,.749}
\definecolor{antiquewhite}{rgb}{1,.996,.96}
\definecolor{pearl}{rgb}{.991,.991,.995}
\let\le\leqslant\let\leq\leqslant
\let\ge\geqslant\let\geq\geqslant
\newenvironment{protocol}%
  {\begin{mdframed}[backgroundcolor=pearl]
    \setlength{\tabcolsep}{3pt}
    \centering
    \begin{tabular}{p{.35\textwidth}cp{.35\textwidth}}
      \iflong
      \multicolumn{1}{c}{\includegraphics[width=1.25cm]{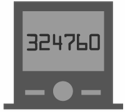}}
      && \multicolumn{1}{c}{\includegraphics[width=1.25cm]{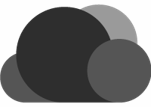}}\\
      \fi%
  }%
  {\end{tabular}\end{mdframed}}
\iflong\fi
\makeatletter
\newcommand*{\noparbreak}{\par\nobreak\@afterheading\smallskip}
\makeatother

\begin{document}

\maketitle
\iflong\thispagestyle{plain}\fi

\begin{abstract}
  The foreseen growing role of outsourced machine learning services is
  raising concerns about the privacy of user data.
  \iflong
  Several technical solutions are being proposed to address the
  issue. Hardware security modules in cloud data centres appear
  limited to enterprise customers due to their complexity, while
  general multi-party computation techniques require a large number of
  message exchanges.
  \fi
  This paper proposes a variety of protocols for privacy-preserving
  regression and classification that (i)~only require additively
  homomorphic encryption algorithms, (ii)~limit interactions to a mere
  request and response, and (iii)~that can be used directly for
  important machine-learning algorithms such as logistic regression
  and SVM classification. The basic protocols are then extended and
  applied to feed-forward neural networks.  

\end{abstract}

\section{Introduction}\label{sec:introduction}

The popularity and hype around machine learning, combined with the
explosive growth of user-generated data, is pushing the development of
\emph{machine learning as a service} (MLaaS). A typical application
scenario of MLaaS is shown in \cref{fig:general}. It involves a
client sending data to a service provider (server) owning and running
a trained machine learning model for a given task (e.g., medical
diagnosis). Both the input data and the model should be kept private:
for obvious privacy reasons on the client's side and to protect
intellectual property on the server's side.

\begin{figure}[ht]
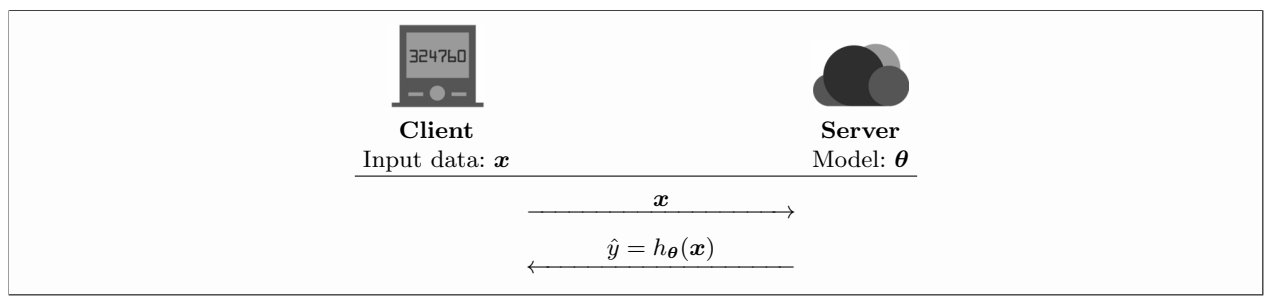

  \begin{protocol}
    \multicolumn{1}{c}{\textbf{Client}}
    && \multicolumn{1}{c}{\textbf{Server}}\\
    \multicolumn{1}{c}{Input data: $\vec{x}$}
    && \multicolumn{1}{c}{Model: $\vec{\theta}$}\\
    \hline
    \noalign{\medskip}
    & $\rarrow[\vec{x}]$ &\\
    \noalign{\smallskip}
    & $\larrow[\yhat = h_{\vec{\theta}}(\vec{x})]$
  \end{protocol}
  \caption{A server offering MLaaS owns a model defined
    by its parameters~$\vec{\theta}$. A client needs the prediction
    $h_{\vec{\theta}}(\vec{x})$ of this model for a new input data
    $\vec{x}$. This prediction is a function of the model and of the
    data.}\label{fig:general}
\end{figure}

In this paper we look at various protocols allowing the realisation of
such scenario in a minimum number of message exchanges between both
parties. Our assumption is that both the client and the server are
honest but curious, that is, they both follow the protocol but may
record information all along with the aim, respectively, to learn the
model and to breach the client's privacy. Our design is guided by the
following \emph{ideal} requirements, in decreasing
importance:\noparbreak
\begin{enumerate}
\item\label{itm:data_conf}\textbf{Input confidentiality}---The server
  does not learn anything about the input data $\vec{x}$ provided by
  the client;
\item\label{itm:output_conf}\textbf{Output confidentiality}---The
  server does not learn the outcome $\hat{y}$
  of the calculation; 
\item \label{itm:min_output}\textbf{Minimal model leakage}---The
  client does not learn any other information about the model beyond
  what is revealed by the successive outputs.
\end{enumerate}

With respect to the issue of model leakage, it is noted that the
client gets access to the outcome, i.e., the value of  
$h_{\vec{\theta}}(\vec{x})$, which may leak information 
about~$\vec{\theta}$, violating Requirement~\ref{itm:min_output}. This is 
unavoidable and not considered as an attack within our framework. 
Possible countermeasures to limit the leakage on the model include
rounding the output or adding some noise to it~\cite{TZJ+16}.

\subsubsection*{Related Work}
Earliest works for private machine learning evaluation~\cite{AS00,LP00} 
were concerned with training models in a privacy-preserving manner. More
recent implementations for linear regression, logistic regression, as 
well as neural networks are offered by SecureML~\cite{MZ17}. The case of 
support vector machines (SVM) is, for example, covered in \cite{ZWY+17}. 
On the contrary, this paper deals with the problem of privately
\emph{evaluating} a linear machine-learning model, including
linear/logistic regression and SVM classification. In~\cite{BLN14},
Bos et al. suggest to evaluate a logistic regression model by
replacing the sigmoid function with its Taylor series expansion. They
then apply fully homomorphic encryption so as to get the output result
through a series of multiplications and additions over encrypted data.
They observe that using terms up to degree $7$ the Taylor expansion
gives roughly two digits of accuracy to the right decimal. Kim et
al.~\cite{KSW+18} argue that such an expansion does not provide enough
accuracy on real-world data sets and propose another polynomial
approximation. For SVM classification, Zhang et
al.~\cite[Protocol~2]{ZWY+17} propose to return an encryption of the
raw output. The client decrypts it and applies the
discriminating function to obtain the corresponding class.
Unfortunately, this leaks more information than necessary on the
model. A similar path is taken by Barni et al. in~\cite{BOP06} for
feed-forward neural networks. Extracting the model (even partially)
is nevertheless more difficult in their case because of the inherent
complexity of the model. Moreover, to further obfuscate it (and
thereby limit the potential leakage), the authors suggest to randomly
permute computing units (neurons) sharing the same activation function
or to add dummy ones. For classification, the approach put forward by
Bost et al.~\cite{BPTG15} is closest to ours. They construct three
classification protocols fulfilling our design criteria
(Requirements~\ref{itm:data_conf}--\ref{itm:min_output}):  hyperplane
decision, na\"{\i}ve Bayes, and decision trees. An approach
orthogonal to ours that introduces privacy in regression or
classification is differential privacy~\cite{DF18}. Crucially, it can
be combined with secure computation, in our case by incorporating
noise in the input vectors or in the model parameters. Differential
privacy can thus be used to enhance the privacy properties of our
protocols.

\subsubsection*{Our Contributions}
Our paper follows the line of work by Bost et al., making use only of
\emph{additively} homomorphic encryption (i.e., homomorphic encryption
supporting additions). We devise new privacy-preserving protocols for
a variety of important prediction tasks. The protocols we propose
either improve on~\cite{BPTG15} or address machine-learning models not
covered in~\cite{BPTG15}. In particular, we aim at minimising the
number of message exchanges to a mere request and response. This is
important when latency is critical. An application
of~\cite[Protocol~4]{BPTG15} to binary SVM classification adds a
round-trip to the comparison protocol whereas our implementation
optimally only needs a \emph{single} round-trip, all included.
Likewise, a single round-trip is needed in our private logistic
regression protocol, as in~\cite{BLN14,ZWY+17}. But contrary
to~\cite{BLN14,ZWY+17}, the resulting prediction is \emph{exact} in
our case (i.e., there is no loss of accuracy) and does not require the
power of fully homomorphic encryption. With respect to neural
networks, we adapt our protocols to binarised networks and to networks
relying of the popular $\ReLU$ activation; see
\Cref{subsec:relu_activation}. As far as we know, this results in the
first privacy-preserving implementation of the non-linear $\ReLU$
function from additively homomorphic encryption.

\iflong
\subsubsection*{Organisation}
The rest of this paper is organised as follows. In \cref{sec:preliminaries},
we give a short summary of important machine learning techniques for
which we will propose secure protocols. We also recall cryptographic
tools on which we will build out protocols. In \cref{sec:pinfer}, we
propose three families of protocols for private inference. They do not
depend on a particular additively homomorphic encryption scheme. We
next apply in \cref{sec:application} our protocols to the private
evaluation of neural networks. Finally, the paper concludes in
\cref{sec:conclusion}.
\fi

\section{Preliminaries}\label{sec:preliminaries}
\iflong
This section reviews some important machine learning models, which all
rely on the computation of an inner product. It also introduces
building blocks that are necessary in the subsequent design of our
privacy-preserving protocols.
\fi

\subsection{Linear Models and Beyond}\label{subsec:linear_models}
Owing to their simplicity, linear models (see, e.g.,
\cite[Chapter~3]{AML12} or~\cite[Chapters~3 and~4]{HTF09}) should not
be overlooked:  They are powerful tools for a variety of machine
learning tasks and find numerous applications.

\subsubsection{Problem Setup}
In a nutshell, machine learning works as follows. Each particular
problem instance is characterised by a set of $d$~features which may
be viewed as a vector $(x_1, \dots, x_d)^\tran$ of $\bbbr^d$. For
practical reasons, a fixed coordinate $x_0 = 1$ is added. We let
$\X \subseteq \{1\} \times \bbbr^d$ denote the input space and $\Y$
the output~space.
\iflong
Integer $d$ is called the dimensionality of the input data.
\fi

There are two phases:
\begin{itemize}
\item The \emph{learning phase} consists in approximating a target
  function $f \colon \X \to \Y$ from
  $\D = \bigl\{(\vec{x_i}, y_i) \in \X \times \Y \mid y_i =
  f(\vec{x_i})\bigr\}_{1 \le i \le n}$, a training set of $n$ pairs of
  elements. Note that the target function can be noisy. The output of
  the learning phase is a function $h_{\vec{\theta}} \colon \X \to \Y$
  drawn from some hypothesis set of functions.
\item In the \emph{testing phase}, when a new data point
  $\vec{x} \in \X$ comes in, it is evaluated on $h_{\vec{\theta}}$ as
  $\yhat = h_{\vec{\theta}}(\vec{x})$. The hat on $y$
  indicates that it is a predicted value.
\end{itemize}

\noindent Since $h_{\vec{\theta}}$ was chosen in a way to ``best
match'' $f$ on the training set $\D$, it is expected that it will
provide a good approximation on a new data point. Namely, we have
$h_{\vec{\theta}}(\vec{x_i}) \approx y_i$ for all $(\vec{x_i}, y_i)
\in \D$ and we should have $h_{\vec{\theta}}(\vec{x}) \approx 
f(\vec{x})$ for $(\vec{x}, \cdot) \notin \D$. Of course, this highly 
depends on the problem under consideration, the data points, and the 
hypothesis set of functions.\par\medskip

In particular, linear models for machine learning use a hypothesis set
of functions of the form $h_{\vec{\theta}}(\vec{x}) = g(\thetax)$
where $\vec{\theta} = (\theta_0, \theta_1, \dots, \theta_d)^\tran \in
\bbbr^{d + 1}$ are the model parameters and $g \colon \bbbr \to \Y$ is
a function mapping the linear calculation to the output space.
\iflong%
\par When the range of $g$ is real-valued and thus the prediction
result $\yhat \in \Y$ is a continuous value (e.g., a quantity or a
probability), we talk about \emph{regression}. When the prediction
result is a discrete value (e.g., a label), we talk about
\emph{classification}. An important sub-case is $\Y = \{+1, -1\}$.
Specific choices for $g$ are discussed in the next sections.
\fi

\subsubsection{Linear Regression}
A \emph{linear regression} model assumes that the real-valued target
function~$f$ is linear---or more generally affine---in the input
variables. In other words, it is based on the premise that $f$ is well
approximated by an affine map; i.e., $g$ is the identity map:
$f(\vec{x_i}) \approx g(\vec{\theta}^\tran\vec{x_i}) =
\vec{\theta}^\tran \vec{x_i}$, $1 \le i \le n$, for training data
$\vec{x_i} \in \X$ and weight vector $\vec{\theta} \in \bbbr^{d+1}$.
\iflong%
This vector $\vec{\theta}$ is interesting as it reveals how the output
depends on the input variables. In particular, the sign of a
coefficient $\theta_j$ indicates either a positive or a negative
contribution to the output, while its magnitude captures the relative
importance of this contribution.
\fi

\iflong%
The linear regression algorithm relies on the least squares method to
find the coefficients of~$\vec{\theta}$: it minimises the sum of
squared errors
$\sum_{i=1}^n \left(f(\vec{x_i}) -
  \vec{\theta}^\tran\vec{x_i}\right)^2$.
\fi
Once $\vec{\theta}$ has been computed, it can be used to produce
estimates on new data points $\vec{x} \in \X$ as $\yhat = \thetax$.

\subsubsection{Support Vector Machines}
\iflong
We now turn our attention to another important problem: how to
classify data into different classes.
\fi
This corresponds to a target function $f$ whose range $\Y$ is
discrete. Of particular interest is the case of two classes, say $+1$
and $-1$, in which case $\Y = \{+1, -1\}$. \iflong Think for example
of a binary decision problem where $+1$ corresponds to a positive
answer and $-1$ to a negative answer.\par\medskip \fi

In dimension $d$, an hyperplane $\Pi$ is given by an equation of the
form
$\theta_0 + \theta_1 X_1 + \theta_2 X_2 + \cdots + \theta_d X_d = 0$
where $\vec{\theta}' = (\theta_1, \dots, \theta_d)^\tran$ is the
normal vector to $\Pi$ and $\theta_0 /\|\vec{\theta}'\|$ indicates the
offset from the origin. \iflong\par\fi When the training data are
\emph{linearly separable}, there is some hyperplane $\Pi$ such that
for each $(\vec{x_i}, y_i) \in \D$, one has
\iflong
\begin{equation}\label{eq:threshold}
  \begin{cases}
    \theta_0 + \theta_1 x_{i,1} + \theta_2 x_{i,2} + \cdots + \theta_n
    x_{i,d} > 0 & \text{if $y_i = +1$}\\
    \theta_0 + \theta_1 x_{i,1} + \theta_2 x_{i,2} + \cdots + \theta_n
    x_{i,d} < 0 & \text{if $y_i = -1$}
  \end{cases}\thinspace,\quad 1 \leq i \leq n\thickspace,
\end{equation}
or equivalently (by scaling $\vec{\theta}$ appropriately):
\[ 
  y_i \, \vec{\theta}^\tran \vec{x_i} \geq 1\thinspace,
  \quad (1 \leq i \leq n) \enspace.
\]
\fi
The training data points $\vec{x_i}$ satisfying
$y_i \, \vec{\theta}^\tran \vec{x_i} = 1$ are called \emph{support
  vectors}.
\iflong\par\fi
When the training data are not linearly separable, it is not possible
to satisfy the previous hard constraint
$y_i \, \vec{\theta}^\tran \vec{x_i} \geq 1$, ($1 \leq i \leq
n$). So-called ``slack variables''
$\xi_i = \max(0, 1 - y_i\, \vec{\theta}^\tran \vec{x_i})$ are
generally introduced in the optimisation problem.
\iflong
They tell how large a violation of the hard constraint there is on 
each training point---note that $\xi_i = 0$ whenever
$y_i \, \vec{\theta}^\tran \vec{x_i} \geq 1$.
\fi

There are many possible choices for $\vec{\theta}$. For better
classification, the separating hyperplane $\Pi$ is chosen so as to
maximise the \emph{margin}; namely, the minimal distance between any
training data point and $\Pi$.
\iflong\par\fi
Now, from the resulting model~$\vec{\theta}$, when a new data point
$\vec{x}$ comes in, its class is estimated as the sign of the
discriminating function $\thetax$; i.e., $\yhat = \sign(\thetax)$.
\iflong Compare with \cref{eq:threshold}.\par\fi

When there are more than two classes, the optimisation problem returns
several vectors $\vec{\theta_k}$, each defining a boundary between a
particular class and all the others. The classification problem
becomes an iteration to find out which $\vec{\theta_k}$ maximises
$\vec{\theta_k}^{\!\tran} \vec{x}$ for a given test point~$\vec{x}$.

\subsubsection{Logistic Regression}
\emph{Logistic regression} is widely used in predictive analysis to
output a probability of occurrence. The logistic function is defined
by the sigmoid function
\iflong \begin{equation} \else $ \fi
  \sigmoid\colon \bbbr \to [0,1],\ t \mapsto \sigmoid(t) =
  \frac{1}{1 + e^{-t}}
\iflong \enspace.\end{equation} \else $. \fi
The logistic regression model returns $h_{\vec{\theta}}(\vec{x}) =
\sigmoid(\thetax) \in [0,1]$, which can be interpreted as the probability
that $\vec{x}$ belongs to the class $y = +1$. The SVM classifier
thresholds the value of $\thetax$ around $0$, assigning to $\vec{x}$ the
class $y= +1$ if $\thetax > 0$ and the class $y = -1$
if $\thetax < 0$. In this respect, the logistic
function is seen as a soft threshold as opposed to the hard threshold,
$+1$ or $-1$, offered by SVM. Other threshold functions are possible.
Another popular soft threshold relies on $\tanh$, the hyperbolic
tangent function, whose output range is $[-1,1]$.

\begin{remark}
  Because the logistic regression algorithm predicts probabilities
  rather than just classes, we fit it through likelihood optimisation.
  Specifically, given the training set $\D$, we learn the model by
  maximising $\prod_{y_i = +1} p_i \cdot \prod_{y_i = -1} (1-p_i)$
  where $p_i = \sigmoid(\vec{\theta}^\tran\vec{x_i})$. This deviates
  from the general description of our problem setup, where the
  learning is directly done on the pairs $(\vec{x_i}, y_i)$. However,
  the testing phase is unchanged: the outcome is expressed as
  $h_{\vec{\theta}}(\vec{x}) = \sigmoid(\vec{\theta}^\tran\vec{x})$.
  It therefore fits our framework for private inference, that is, the
  private evaluation of
  $h_{\vec{\theta}}(\vec{x}) = g(\vec{\theta}^\tran\vec{x})$ for a
  certain function $g$; the sigmoid function $\sigmoid$ in this case.
\end{remark}

\subsection{Cryptographic Tools}\label{subsec:cryptotools}

\subsubsection{Representing Real Numbers}\label{subsubsec:real_numbers}
So far, we have discussed a number of machine learning models using
real numbers, but the cryptographic tools we intend to use require
working on integers.
\iflong%
We therefore start by recalling the necessary conversion.
An encryption algorithm takes as input an encryption key and a
plaintext message and returns a ciphertext. We let $\M \subset \bbbz$
denote the set of messages that can be encrypted.
\fi
In order to operate over encrypted data, we need to accurately 
represent real numbers as elements of
\iflong%
$\M$ (i.e., a finite subset of $\bbbz$).
\fi

To do that, since all input variables of machine learning models are 
typically rescaled in the range $[-1,1]$, one could use a fixed point 
representation. A real number $x$ with a fractional part of at most 
$P$ bits uniquely corresponds to signed integer $z = x \cdot 2^P$. 
Hence, with a fixed-point representation, a real number $x$ is 
represented by
\iflong
\[
  z = \lfloor x \cdot 2^P \rfloor\thinspace,
\]
\fi
where integer $P$ is called the bit-precision. The sum of
$x_1, x_2 \in \bbbr$ is performed as $z_1 + z_2$ and their 
multiplication as $\lfloor (z_1 \cdot z_2)/2^P \rfloor$.

\subsubsection{Additively Homomorphic Encryption}

\iflong%
Homomorphic encryption schemes come in different flavours. Before
Gentry's breakthrough result~(\cite{Gen09}), only addition operations
or multiplication operations on ciphertexts---but not both---were
supported. Schemes that can support an arbitrary number of additions
and of multiplications are termed \emph{fully} homomorphic encryption
(FHE) schemes. Our privacy-preserving protocols only need an 
\emph{additively} homomorphic encryption scheme.
The minimal security notion that we require is semantic
security~\cite{GM84}; in particular, encryption is
probabilistic. 
\fi

It is useful to introduce 
some notation. We let $\HEnc{\cdot}$ and $\HDec{\cdot}$ denote the
encryption and decryption algorithms, respectively. The message space
is an additive group $\M \cong \bbbz/M\bbbz$. It consists of integers
modulo $M$ and we view it as
$\M = \{-\lfloor M/2 \rfloor, \dots, \lceil M/2 \rceil - 1\}$ in order
to keep track of the sign.
\iflong%
The elements of $\M$ are uniquely identified with $\bbbz/M\bbbz$ via
the mapping $\Upsilon\colon \M \simto \bbbz/M\bbbz$,
${m \mapsto m \bmod M}$. The inverse mapping is given by
$\Upsilon^{-1}\colon \bbbz/M\bbbz \simto \M, m \mapsto m$ if
$m < \lceil M/2 \rceil$ and $m \mapsto m-M$ otherwise.
\fi
Ciphertexts are noted with Gothic letters.
\iflong%
The encryption of a message $m \in \M$ is
obtained using public key $\pk$ as $\ct{m} = \HEnc{m}_\pk$. It is
then decrypted using the matching secret key $\sk$ as
$m = \HDec{\ct{m}}_\sk$. When clear from the context, we drop the
$\pk$ or $\sk$ subscripts and sometimes use $\HHEnc{\cdot}$ and
$\HHDec{\cdot}$ to denote another encryption algorithm. 
\fi
If $\vec{m} = (m_1, \dots, m_d) \in \M^d$ is a vector, we write
$\vec{\ct{m}} = \HEnc{\vec{m}}$ as a shorthand for
$(\ct{m}_1, \dots, \ct{m}_d) = \left(\HEnc{m_1}, \dots,
  \HEnc{m_d}\right)$.

Algorithm $\HEnc{\cdot}$ being \iflong additively \fi homomorphic (over $\M$)
means that given any two plaintext messages $m_1$ and $m_2$ and their
corresponding ciphertexts ${\ct{m}_1 = \HEnc{m_1}}$ and
$\ct{m}_2 = \HEnc{m_2}$, we have
$\ct{m}_1 \boxplus \ct{m}_2 = \HEnc{m_1 + m_2}$ and
$\ct{m}_1 \boxminus \ct{m}_2 = \HEnc{m_1 - m_2}$ for some publicly
known operations $\boxplus$ and $\boxminus$ on ciphertexts. By
induction, for a given integer scalar $r \in \bbbz$, we also have
\iflong
\begin{align*}
  \HEnc{r\cdot m_1}
  &= \HEnc{m_1 + \dots + m_1}
    = \HEnc{m_1} \boxplus \dots \boxplus \HEnc{m_1}\\
  &= \underbrace{\ct{m}_1 \boxplus \dots \boxplus \ct{m}_1}_{\text{$r$ times}}
    \coloneqq r \odot \ct{m}_1\enspace.
\end{align*}
\fi
It is worth noting here that the decryption of
$(\ct{m}_1 \boxplus \ct{m}_2)$ gives $(m_1 + m_2)$ as an element of
$\M$; that is,
$\HDec{\ct{m}_1 \boxplus \ct{m}_2} \equiv m_1 + m_2 \pmod M$.
Similarly, we also have
$\HDec{\ct{m}_1 \boxminus \ct{m}_2} \equiv m_1 - m_2 \pmod M$ and
$\HDec{r \odot \ct{m}_1} \equiv r\cdot m_1 \pmod M$.

\subsubsection{Private Comparison Protocol}
In~\cite{DGK08,DGK09}, Damg{\aa}rd et al. present a
protocol for comparing private values. It was later extended and
improved in~\cite{EFG+09} and~\cite{Veu12,JS18}. The protocol makes
use of an additively homomorphic encryption scheme. It compares two
non-negative $\ell$-bit integers. The message space is
$\M \cong \bbbz/M\bbbz$ with $M \ge 2^\ell$ and is supposed to behave
like an integral domain\iflong\ (for example, $M$ a prime or an
RSA-type modulus).\fi

\paragraph*{DGK+ protocol}
\iflong
The setting is as follows.
\fi
A client possesses a private $\ell$-bit value
$\mu = \sum_{i=0}^{\ell-1} \mu_i\,2^i$ and a server possesses a
private $\ell$-bit value $\eta = \sum_{i=0}^{\ell-1} \eta_i\,2^i$. They
seek to respectively obtain bits $\deltaC$ and
$\deltaS$ such that $\deltaC \oplus \deltaS = \pred{\mu \le \eta}$
(where $\oplus$ represents the exclusive \textsc{or} operator, and
$\pred{\mathsf{Pred}}=1$ if predicate $\mathsf{Pred}$ is true, and $0$
otherwise). Following~\cite[Fig.~1]{JS18}, the DGK+ protocol proceeds
in four steps:\noparbreak
\begin{enumerate}
\item The client encrypts each bit $\mu_i$ of $\mu$ under its public
  key and sends $\HEnc{\mu_i}$, $0 \le i \le \ell-1$, to the server.
\item The server chooses
  uniformly at random a bit
  $\deltaS$ and defines $s = 1 - 2\deltaS$. It also
  selects $\ell+1$  random non-zero scalars $r_i \in \M$,
  $-1 \le i \le \ell-1$.
\item\label{itm:hi*} Next, the server computes\footnote{Given
    $\HEnc{\mu_i}$, the server obtains $\HEnc{\eta_i \oplus \mu_i}$ as
    $\HEnc{\mu_i}$ if $\eta_i = 0$, and as
    $\HEnc{1} \boxminus \HEnc{\mu_i}$ if
    $\eta_i = 1$.\label{footnote:encrypt_xor}}
  \begin{equation}\label{eq:hi*}
    \begin{cases}
      \HEnc{h_i^*}      
      = r_i \odot \bigl(\HEnc{1} \boxplus \HEnc{s\cdot\mu_i} \boxminus
      \HEnc{s\cdot\eta_i}  \boxplus (\textstyle\boxsum_{j=i+1}^{\ell-1}
      \HEnc{\mu_j \oplus \eta_j})\bigr)\\
      &\qquad\mathllap{\text{for $\ell -1 \ge i \ge 0$}\thinspace,}\\
      \HEnc{h_{-1}^*}
      = r_{-1} \odot \bigl(\HEnc{\deltaS} \boxplus
      \textstyle\boxsum_{j=0}^{\ell-1} \HEnc{\mu_j \oplus
        \eta_j}\bigr)
    \end{cases}
  \end{equation}
  and sends the $\ell+1$ ciphertexts $\HEnc{h_i^*}$ in a \emph{random
    order} to the client.
\item Using its private key, the client decrypts the received
  $\HEnc{h_i^*}$'s. If one is decrypted to zero, the client sets
  $\deltaC = 1$. Otherwise, it sets $\deltaC = 0$.
\end{enumerate}

\begin{remark}\label{rem:DGKsec}
  At this point, neither the client, nor the server, knows whether
  $\mu \leq \eta$ holds. One of them (or both) needs to reveal its
  share of $\delta$ ($= \deltaC \oplus \deltaS$) so that the other can
  find out. Following the original DGK protocol~\cite{DGK08}, this
  modified comparison protocol is secure in the semi-honest model
  (i.e., against honest but curious adversaries).
\end{remark}

\iflong
\paragraph*{Correctness} The correctness of the protocol follows from
the fact that $\mu \leq \eta$ if only and only if:
\begin{itemize}[topsep=0pt]
\item $\mu = \eta$, or
\item there exists some index $i$, with $0 \le i \le \ell-1$, such
  that:
  \begin{enumerate}
    \renewcommand{\theenumi}{\roman{enumi}}
  \item $\mu_i < \eta_i$, and
  \item $\mu_j = \eta_j$ for $i + 1 \leq j \leq \ell-1$\thinspace.
  \end{enumerate}
\end{itemize}

\noindent As pointed out in~\cite{DGK08}, when $\mu \neq \eta$, this
latter condition is equivalent to the existence of some index
$i\in [0, \ell-1]$, such that
$\mu_i - \eta_i + 1 + \sum_{j=i+1}^{\ell-1} (\mu_j \oplus \eta_j) =
0$. This test was subsequently replaced in~\cite{EFG+09,JS18} to allow
the secret sharing of the comparison bit across the client and the
server as $\pred{\mu \le \eta} = \deltaC \oplus
\deltaS$. Adapting~\cite{JS18}, the new test checks the existence of
some index $i\in [0, \ell-1]$, such that
\[
  h_i = s(\mu_i - \eta_i) + 1 + \textstyle\sum_{j=i+1}^{\ell-1}
  (\mu_j \oplus \eta_j)
\]
is zero. When $\deltaS = 0$ (and thus $s=1$) this occurs if
$\mu < \eta$; when $\deltaS = 1$ ($s=-1$) this occurs if $\mu >
\eta$. As a result, the first case yields
$\deltaS = \neg\pred{\mu < \eta} = 1 \oplus \pred{\mu < \eta}$ while
the second case yields
$\deltaS = \pred{\mu > \eta} = \neg\pred{\mu \le \eta} = 1 \oplus
\pred{\mu \le \eta}$. This discrepancy is corrected in~\cite{Veu12} by
augmenting the set of $h_i$'s with an additional value $h_{-1}$ given
by
$h_{-1} = \deltaS + \textstyle\sum_{j=0}^{\ell-1} (\mu_j \oplus
\eta_j)$. It is worth observing that $h_{-1}$ can only be zero when
$\deltaS = 0$ and $\mu=\eta$. Therefore, in all cases, when there
exists some index $i$, with $-1 \le i \le \ell-1$, such that
$h_i = 0$, we have $\deltaS = 1\oplus\pred{\mu \le \eta}$, or
equivalently, $\pred{\mu \le \eta} = \deltaS \oplus 1$.

It is easily verified that $\HEnc{h_i^*}$ as computed in
Step~\ref{itm:hi*} is the encryption of $r_i\cdot h_i \pmod M$.
Clearly, if $r_i \cdot h_i \pmod M$ is zero then so is $h_i$ since, by
definition, $r_i$ is non-zero---remember that $M$ is chosen such that
$\bbbz/M\bbbz$ acts as an integral domain. Hence, if one of the
$\HEnc{h_i^*}$'s decrypts to $0$ then
$\pred{\mu \le \eta} = \deltaS \oplus 1 = \deltaS \oplus \deltaC$; if
not, one has $\pred{\mu \le \eta} = \deltaS = \deltaS \oplus \deltaC$.
This concludes the proof of correctness.
\fi


\section{Basic Protocols of Privacy-Preserving Inference}\label{sec:pinfer}

In this section, we present three families of protocols for private
inference. They aim to satisfy the ideal requirements given in the
introduction while keeping the number of exchanges to a bare minimum. 
Interestingly, they only make use of additively homomorphic encryption.

We keep the general model presented in the introduction, but now work
with integers only. The client holds
$\vec{x} = (1, x_1, \dots, x_d)^\tran \in \M^{d+1}$, a private feature
vector, and the server possesses a trained machine-learning model
given by its parameter vector
$\vec{\theta} = (\theta_0, \dots, \theta_d)^\tran \in \M^{d+1}$ or, in
the case of feed-forward neural networks a set of matrices made of
such vectors. At the end of protocol, the client obtains the value of
$g(\thetax)$ for some function $g$ and learns nothing else; the server
learns nothing. To make the protocols easier to read, for a
real-valued function~$g$, we abuse notation and write $g(t)$ for an 
integer $t$ assuming $g$ also includes the conversion to real values;
see \cref{subsec:cryptotools}. We also make the distinction between
the encryption algorithm $\HEnc{\cdot}$ using the client's public key
and the encryption algorithm $\HHEnc{\cdot}$ using the server's public
key and stress that, not only keys are different, but the algorithm
could also be different. We use $\HDec{\cdot}$ and $\HHDec{\cdot}$
for the respective corresponding decryption algorithms.

\subsection{Private Linear/Logistic Regression}\label{subsec:plogreg}
\iflong
\subsubsection*{Private Linear Regression}
As seen in \Cref{subsec:linear_models},
linear regression produces estimates
using the identity map for $g$:  $\yhat = \thetax$. Since
$\thetax \iflong = \sum_{j=0}^d \theta_j\,x_j \fi$ is linear, given an encryption
$\HEnc{\vec{x}}$ of $\vec{x}$, the value of
$\HEnc{\vec{\theta}^\tran\vec{x}}$ can be homomorphically evaluated,
in a provably secure way \cite{GLLM04}.

Therefore, the client encrypts its feature vector $\vec{x}$ under its
public key with an additively homomorphic encryption algorithm
$\HEnc{\cdot}$, and sends $\HEnc{\vec{x}}$ to the server. Using
$\vec{\theta}$, the server then computes
$\HEnc{\vec{\theta}^\tran\vec{x}}$ and returns it the client.
Finally, the client uses its private key to decrypt
$\HEnc{\vec{\theta}^\tran\vec{x}} = \HEnc{\yhat}$ and gets the
output~$\yhat$. This is straightforward and only requires one round
of communication.

\subsubsection*{Private Logistic Regression}
\fi
Things get more complicated for logistic regression.
At first sight, it seems counter-intuitive that additively homomorphic
encryption could suffice to evaluate a logistic regression model over
encrypted data. \iflong After all, the sigmoid function,
$\sigmoid(t)$, is non-linear.\par
The key observation is that the sigmoid function is \emph{injective}:
\[
  \sigmoid(t_1) = \sigmoid(t_2) \implies t_1 = t_2\enspace.
\]
\fi 
So the client does not learn more about the model $\vec{\theta}$ from
$t \coloneqq \vec{\theta}^\tran \vec{x}$ than it can learn from
$\yhat \coloneqq \sigmoid(t)$ since the value of $t$ can be recovered
from $\yhat$ using $t = \sigmoid^{-1}(\yhat) = \ln(\yhat/(1-\yhat))$. 
\iflong
Consequently, rather than returning an encryption of the 
prediction~$\yhat$, we let the server return an encryption of $t$, 
without any security loss in doing so.
\fi

\subsubsection{Our Core Protocol}
The protocol we propose for privacy-preserving linear or logistic
regression is detailed in \cref{fig:linear_regression_core}. Let $(\pkC,\skC)$ denote
the client's matching pair of public encryption key/private decryption
key for an additively homomorphic encryption scheme $\HEnc{\cdot}$.
We use the notation of \Cref{subsec:cryptotools}.
If $B$ is an upper bound on
the inner product (in absolute value), the message space
$\M = \{-\lfloor M/2\rfloor, \dots, \lceil M/2\rceil - 1\}$ should be
such that $M \ge 2B + 1$.
\begin{figure}[h]
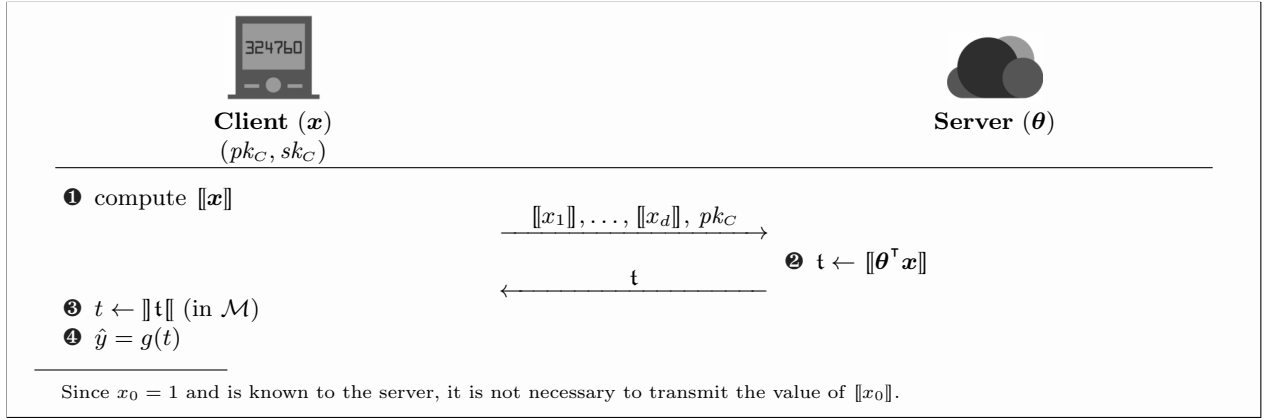

  \begin{protocol}
    \multicolumn{1}{c}{\textbf{Client} ($\vec{x}$)}
    && \multicolumn{1}{c}{\textbf{Server} ($\vec{\theta}$)}\\
    \multicolumn{1}{c}{$(\pkC, \skC)$}
    && \multicolumn{1}{c}{}\\
    \hline
    \noalign{\medskip}
    \cding{1}\enspace compute\footnotetext{\scriptsize Since $x_0 = 1$ and
      is known to the
      server, it is not necessary to transmit the value of
      $\HEnc{x_0}$.} $\HEnc{\vec{x}}$\\[-1ex]
    & $\rarrow[\HEnc{x_1}, \dots, \HEnc{x_d},\, \pkC]$\\
    && \cding{2}\enspace $\ct{t} \gets \HEnc{\vec{\theta}^\tran
      \vec{x}}$\\[-1ex]
    & $\larrow[\ct{t}]$\\[-1ex]
    \cding{3}\enspace $t \gets \HDec{\ct{t}}$ (in $\M$)\\
    \cding{4}\enspace $\yhat = g(t)$
  \end{protocol}
  \caption{Privacy-preserving regression. Encryption is done using the
    client's public key and noted $\HEnc{\cdot}$. The server learns
    nothing. Function $g$ is the identity map for linear regression and
    the sigmoid function for logistic regression.}\label{fig:linear_regression_core}
\end{figure}

\iflong
In more detail, our core protocol goes as follows.
\fi
\begin{enumerate}
\item In a first step, the client encrypts its feature vector
  $\vec{x} \in \M^{d+1}$ under its public key $\pkC$
  and gets
  $\HEnc{\vec{x}} = (\HEnc{x_0}, \HEnc{x_1}, \dots, \HEnc{x_d})$. The
  ciphertext $\HEnc{\vec{x}}$ along with the client's public key are
  sent to the server.
\item In a second step, from its model $\vec{\theta}$, the server 
  computes an encryption of the inner product over encrypted data as:
  \iflong
  \[ 
    \ct{t}
    =\HEnc{\vec{\theta}^\tran\vec{x}}
    = \HEnc{\theta_0} \boxplus \boxsum_{j=1}^d \theta_j \odot
    \HEnc{x_j}\enspace.
  \]
  \fi
  The server returns $\ct{t}$ to the client.
\item In a third step, the client uses its private decryption key  
  $\skC$ to decrypt $\ct{t}$, and gets the inner product
  $t = \thetax$ as a signed integer of $\M$.
\item In a final step, the client applies the $g$ function to obtain   
  the prediction $\yhat$ corresponding to input vector~$\vec{x}$.
\end{enumerate}

\subsubsection{Dual Approach}
The previous protocol encrypts with the client's public
key $\pkC$. In the dual approach, the server's public key is used for
encryption. Let $(\pkS, \skS)$ denote the public/private key pair of
the server for some additively homomorphic encryption scheme
$\left(\HHEnc{\cdot}, \HHDec{\cdot}\right)$. The message space $\M$
is unchanged.

In this case, the server needs to publish an encrypted version
$\HHEnc{\vec{\theta}}$ of its model. The client must therefore get a
copy of $\HHEnc{\vec{\theta}}$ once, but can then engage in the
protocol as many times as it wishes. One could also suppose that each
client receives a different encryption of $\vec{\theta}$ using a
server's encryption key specific to the client, or that a key rotation
is performed on a regular basis.
This protocol uses a mask $\mu$ which is chosen uniformly at random in 
$\M$. Consequently, it is important to see that~$\tstar$
($\equiv \thetax + \mu \pmod{M}$) is also uniformly distributed over 
$\M$. Thus, the server gains no bit of information from~$\tstar$.
The different steps are summarised
in \cref{fig:linear_regression_dual}.

\begin{figure}[h]
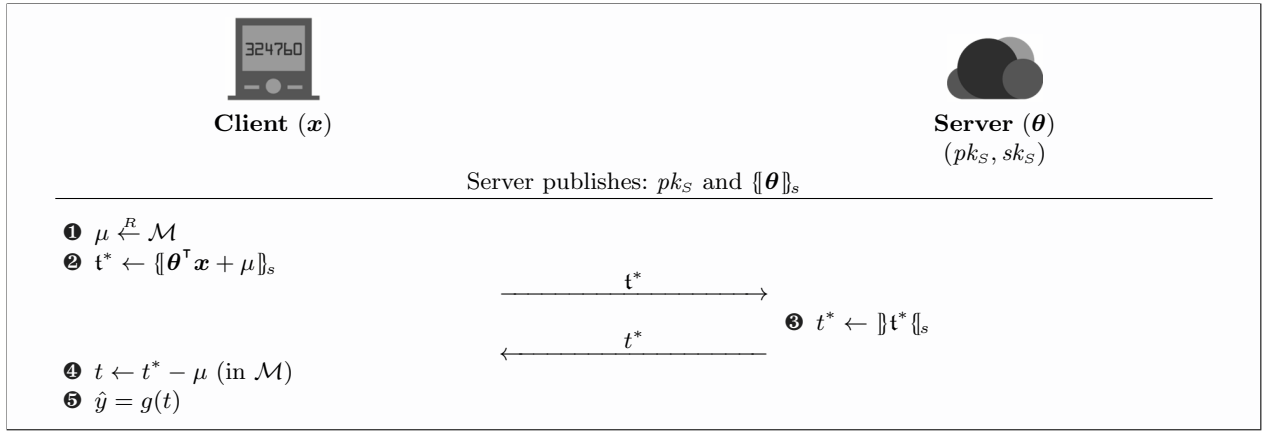

  \begin{protocol}
    \multicolumn{1}{c}{\textbf{Client} ($\vec{x}$)}
    && \multicolumn{1}{c}{\textbf{Server} ($\vec{\theta}$)}\\
    &&\multicolumn{1}{c}{$(\pkS, \skS)$}\\
    & \clap{Server publishes: $\pkS$ and $\HHEnc{\vec{\theta}}$}\\
    \hline
    \noalign{\medskip}
    \cding{1}\enspace $\mu \getsr \M$\\
    \cding{2}\enspace $\ct{t}^* \gets \HHEnc{\thetax +
    \mu}$\\[-1ex]
    & $\rarrow[\ct{t}^*]$\\
    && \cding{3}\enspace $\tstar \gets \HHDec{\ct{t}^*}$\\[-1ex]
    & $\larrow[\tstar]$\\[-1ex]
    \cding{4}\enspace $t \gets \tstar - \mu$ (in $\M$)\\
    \cding{5}\enspace $\yhat = g(t)$
  \end{protocol}
  \caption{Dual approach for privacy-preserving regression. Here,
    encryption is done using the server's public key $\pkS$ and noted
    $\HHEnc{\cdot}$. Function $g$ is the identity map for linear
    regression and the sigmoid function for logistic
    regression.}\label{fig:linear_regression_dual}
\end{figure}

\iflong
\subsubsection{Variant and Extensions}
In a variant, in Step~2 of \cref{fig:linear_regression_core} (resp. Step~3 of
\cref{fig:linear_regression_dual}), the server can add some noise $\epsilon$ by
defining $\ct{t}$ as
$\ct{t} \gets \HEnc{{\vec{\theta}}^\tran\vec{x} + \epsilon} =
\HEnc{\vec{\theta}^\tran\vec{x}} \boxplus \HEnc{\epsilon}$
(resp. $\tstar$ as $\tstar \gets \HHDec{\ct{t}^*} + \epsilon$). This
presents the advantage of limiting the leakage on $\vec{\theta}$
resulting from the output result.
\fi
The proposed methods are not limited to the identity map or the
sigmoid function but generalise to any injective function $g$.
\iflong
This includes the $\tanh$ activation function alluded to in
\cref{subsec:linear_models} where $g(t) = \tanh(t)$, as well as
\[
  \begin{array}{l@{\qquad}l}
    g(t) = \arctan(t) \quad[\text{arctan}]\thinspace,&
    g(t) = t/(1 + \abs{t}) \quad[\text{softsign}]\thinspace,\\
    g(t) = \ln(1 + e^t) \quad[\text{softplus}]\thinspace,&
    g(t) = \begin{cases} 0.01t & \text{for $t < 0$}\\
      t & \text{for $t \ge 0$}\end{cases} \quad[\text{leaky $\ReLU$}]\thinspace,
  \end{array}
\]
\fi 
and more. For any injective function~$g$, there is no more information
leakage in returning $\thetax$ than $g(\thetax)$.

\subsection{Private SVM Classification}\label{subsec:svm_classification}
As discussed in \Cref{subsec:linear_models}, SVM inference can be abridged to
the evaluation of the sign of an inner product. However, the $\sign$
function is clearly not injective. Our idea is to make use of a
privacy-preserving comparison protocol. For concreteness, we consider
the DGK+ protocol; but any privacy-preserving comparison protocol
could be adapted.

\iflong
\subsubsection{A Na{\"{\i}}ve Protocol}
A client holding a private feature vector $\vec{x}$ wishes to evaluate
$\sign(\vec{\theta}^\tran \vec{x})$ where $\vec{\theta}$ parametrises
an SVM classification model. In the primal approach, the client can
encrypt $\vec{x}$ and send $\HEnc{\vec{x}}$ to the server. Next, the
server computes ${\HEnc{\eta} = \HEnc{\thetax + \mu}}$ for some random
mask~$\mu$ and sends~$\HEnc{\eta}$ to the client. The client
decrypts~$\HEnc{\eta}$ and recovers~$\eta$. Finally, the client and
the server engage in a private comparison protocol with respective
inputs $\eta$ and $\mu$, and the client deduces the sign of $\thetax$
from the resulting comparison bit $\pred{\mu \le \eta}$.

There are two issues. If we use the DGK+ protocol for the private
comparison, at least one extra exchange from the server to the client
is needed for the client to get $\pred{\mu \le \eta}$. This can be fixed
by considering the dual approach. A second, more problematic, issue
is that the decryption of $\HEnc{\eta} \coloneqq \HEnc{\thetax + \mu}$
yields $\eta$ as an element of $\M \cong \bbbz/M\bbbz$, which is not
necessarily equivalent to the \emph{integer} $\thetax + \mu$. Note
that if the inner product $\thetax$ can take any value in $\M$,
selecting a smaller value for $\mu \in \M$ to prevent the modular
reduction does not solve the issue because the value of $\eta$ may
then leak information on $\vec{\theta}^\tran \vec{x}$.
\fi

\subsubsection{Our Core Protocol}
\iflong
Instead, we suggest to select the message space much larger than the
upper bound $B$ on the inner product, so that the computation
will take place over the integers. 
\fi
Specifically, if $\thetax \in [-B, B]$ then, letting $\ell$ be the
bit-length of $B$, the message space
$\M = \{-\lfloor M/2\rfloor,\allowbreak \dots, \lceil M/2\rceil -1\}$
\iflong is \fi dimensioned such that
$M \geq 2^\ell(2^{\secpar}+1) - 1$ for some security
parameter~$\secpar$. Let $\mu$ be an $(\ell+\secpar)$\nobreakdash-bit
integer that is chosen such that $\mu \geq B$. By construction we will
then have $0 \leq \thetax + \mu < M$ so that the decrypted value
modulo $M$ corresponds to the actual integer value. As will become
apparent, this presents the further advantage of optimising the
bandwidth requirements: the number of exchanged ciphertexts depends on
the length of $B$ and not on the length of $M$ (notice that
$M = \#\M$).

Our resulting core protocol for private SVM classification of a
feature vector~$\vec{x}$ is illustrated in \cref{fig:svm_classification_core} and 
includes the following steps:
\begin{enumerate}[start=0]
\item The server publishes $\pkS$ and $\HHEnc{\vec{\theta}}$.
\item Let $\secpar$ be a security parameter. The client starts by
  picking uniformly at random in $[2^\ell - 1, 2^{\ell+\secpar})$ an 
  integer $\mu = \sum_{i=0}^{\ell+\secpar-1} \mu_i\,2^i$ .
\item In a second step, the client computes, over encrypted data, the
  inner product $\thetax$ and masks the result with  $\mu$ to get
  \iflong
    \[
    \ct{t}^*
    = \HHEnc{\theta_0} \boxplus
    \biggl(\boxsum_{j=1}^d x_j \odot
    \HHEnc{\theta_j}\biggr) \boxplus \HHEnc{\mu}\enspace.
  \]
  \fi

\item Next, the client individually encrypts the first $\ell$ bits of
  $\mu$ with its own encryption key to get $\HEnc{\mu_i}$, for
  $0 \le i \le \ell -1$, and sends $\ct{t}^*$ and the $\HEnc{\mu_i}$'s
  to the server.
\item Upon reception, the server decrypts $\ct{t}^*$ to get
  $\tstar \coloneqq \HHDec{\ct{t}^*} \bmod M = \vec{\theta}^\tran
  \vec{x} + \mu$ and defines the $\ell$-bit integer
  $\eta \coloneqq \tstar \bmod 2^\ell$.
\item The DGK+ protocol is now
  applied to two $\ell$-bit values
  $\ul{\mu} \coloneqq \mu \bmod 2^\ell =
  \sum_{i=0}^{\ell-1}\mu_i\,2^i$ and
  $\eta = \sum_{i=0}^{\ell-1}\eta_i\,2^i$. The server selects the
  $(\ell+1)$-th bit of $\tstar$ for $\deltaS$ (i.e.,
  $\deltaS = \lfloor \tstar/2^\ell\rfloor \bmod 2$), defines
  $s = 1 - 2\deltaS$, and forms the $\HEnc{h_i^*}$'s (with
  $-1 \le i \le \ell-1$) as defined by \cref{eq:hi*}. The server
  permutes randomly the $\HEnc{h_i^*}$'s and sends them to the client.
\item The client decrypts the $\HEnc{h_i^*}$'s and gets the $h_i^*$'s.
  If one of them is zero, it sets $\deltaC = 1$;
  otherwise it sets $\deltaC = 0$.
\item As a final step, the client obtains the predicted class as
  $\yhat = (-1)^{\neg(\deltaC \oplus \mu_\ell)}$, where $\mu_\ell$ denotes
  bit number~$\ell$ of $\mu$.
\end{enumerate}

\iflong
\begin{figure}
\else
\begin{figure}[hbt]
\fi
  \begin{protocol}
    \multicolumn{1}{c}{\textbf{Client} ($\vec{x}$)}
    && \multicolumn{1}{c}{\textbf{Server} ($\vec{\theta}$)}\\
    \multicolumn{1}{c}{$(\pkC, \skC)$}
    && \multicolumn{1}{c}{$(\pkS, \skS)$}\\
    & \clap{Server publishes: $\pkS$ and $\HHEnc{\vec{\theta}}$}\\
    \hline
    \noalign{\medskip}
    \cding{1}\enspace $\mu
    \begin{array}[t]{@{}l}
      {}\getsr [2^\ell-1,2^{\ell+\secpar})
      \iflong \\
      {}= \sum_{i=0}^{\ell+\secpar-1}\mu_i\,2^i
      \fi 
    \end{array}$\\
    \cding{2}\enspace $\ct{t}^* \gets
    \HHEnc{\thetax + \mu}$\\
    \cding{3}\enspace
    \begin{tabular}[t]{@{}l@{}}
      compute $\HEnc{\mu_i}$\iflong, \fi
      for $0 \leq i \leq \ell - 1$
    \end{tabular}\\[-1ex]
    & $\rarrow[\ct{t}^*,\, \HEnc{\mu_0}, \dots, \HEnc{\mu_{\ell-1}}]$\\[-1ex]
    &&\cding{4}\enspace
    \begin{tabular}[t]{@{}l}
      \textbullet\enspace  $t^* \gets \mathrlap{\HHDec{\ct{t}^*} \bmod M}$\\
      \textbullet\enspace  $\eta
          \begin{array}[t]{@{}l}
            {}\gets t^* \bmod 2^\ell
            \iflong \\
            {}= \sum_{i=0}^{\ell-1} \eta_i\,2^i
            \fi 
          \end{array}$
        \end{tabular}\\
    && \cding{5}\enspace
       \begin{tabular}[t]{@{}l@{}}
         \textbullet\enspace $\deltaS \gets \mathrlap{\lfloor
         t^*/2^\ell \rfloor \bmod 2}$\\ 
         \textbullet\enspace compute $\HEnc{h_i^*}$\iflong, \fi
         for $\mathrlap{-1 \le i \le \ell-1}$          
        \end{tabular}\\[-1ex]
    & $\larrow[\{\HEnc{h_{-1}^*}, \dots, \HEnc{h_{\ell-1}^*}\}\\
      \text{in random order}]$\\[-1ex]
    \cding{6}\enspace $\deltaC \gets \pred[\bigl]{\exists i \mid h_i^* = 0}$\\
    \cding{7}\enspace
    $\hat{y} = (-1)^{\neg(\deltaC \oplus \mu_\ell)}$
  \end{protocol}
  \caption{Privacy-preserving SVM classification.
    Note that some data is encrypted
    using the client's public key~$\pkC$, while other, is encrypted
    using the server's public key~$\pkS$. They are noted
    $\HEnc{\cdot}$ and $\HHEnc{\cdot}$
    respectively.}\label{fig:svm_classification_core}
\end{figure}

Again, the proposed protocol keeps the number of interactions between
the client and the server to a minimum: a request and a response.

\paragraph*{Correctness}
\iflong
To prove the correctness, we need the two following simple lemmata.
\begin{lemma}\label{lem:lemma1}
  Let $a$ and $b$ be two non-negative integers. Then for any positive
  integer $n$,
  $\lfloor(a-b)/n\rfloor = \lfloor a/n \rfloor - \lfloor b/n \rfloor +
  \lfloor ((a \bmod n) - (b \bmod n))/n \rfloor$.
\end{lemma}

\begin{proof}
  Write $a = \bigl\lfloor \frac{a}{n} \bigr\rfloor n + (a \bmod n)$
  and $b = \bigl\lfloor \frac{b}{n} \bigr\rfloor n + (b \bmod n)$.
  Then
  $a - b = \bigl(\bigl\lfloor \frac{a}{n} \bigr\rfloor - \bigl\lfloor
  \frac{b}{n} \bigr\rfloor\bigr) n + ( a \bmod n) - (b \bmod
  n)$. Recalling that for $n_0 \in \bbbz$ and $x \in \bbbr$,
  $\lfloor x + n_0 \rfloor = \lfloor x \rfloor + n_0$ and
  $\lfloor - n_0 \rfloor = - \lfloor n_0 \rfloor$, the lemma follows
  by integer division through $n$.\qed
\end{proof}

\begin{lemma}\label{lem:lemma2}
  Let $a$ and $b$ be two non-negative integers smaller than some
  positive integer $n$. Then
  $\pred{b \le a} = 1 + \lfloor (a-b)/n \rfloor$.
\end{lemma}

\begin{proof}
  By definition $0 \leq a < n$ and $0 \leq b < n$. If $b \le a$ then
  $0 \le \frac{a - b}{n} < 1$ and thus
  $\bigl\lfloor \frac{a-b}{n} \bigr\rfloor = 0$; otherwise, if $b > a$
  then $-1 < \frac{a - b}{n} < 0$ and so
  $\bigl\lfloor \frac{a-b}{n} \bigr\rfloor = -1$.\qed
\end{proof}
\fi

Remember that, by construction,
$\vec{\theta}^\tran \vec{x} \in [-B,B]$ with $B=2^\ell - 1$, that
$\mu \in [2^\ell-1,2^{\ell+\secpar})$, and by definition that
$\tstar \coloneqq \HHDec{\ct{t}^*} \bmod M$ with
$\ct{t}^* = \HHEnc{\thetax + \mu}$. Hence, in
Step~4, the server gets
$\tstar = \vec{\theta}^\tran \vec{x} + \mu \bmod M = \vec{\theta}^\tran
\vec{x} + \mu$ (over $\bbbz$) since
$0 \le \vec{\theta}^\tran \vec{x} + \mu \le 2^{\ell} - 1 +
2^{\ell+\secpar} - 1 < M$.
Let $\delta \coloneqq \deltaC \oplus \deltaS = \pred{\ul{\mu} \le \eta}$
(with $\ul{\mu} := \mu \bmod 2^\ell$ and $\eta := t^* \bmod{2^\ell}$) 
denote the result of the private comparison in Steps~5 and 6 with the
DGK+ protocol. 

Either of those two conditions holds true
\[
  \left\lbrace
  \begin{array}{@{}rll}
    0 \le &\thetax < 2^\ell
    &\iff 1 \le  \frac{\thetax + 2^\ell}{2^\ell} < 2
    \iff \bigl\lfloor \frac{\thetax +
      2^\ell}{2^\ell} \bigr\rfloor = 1\\ 
    -2^\ell < &\thetax < 0
    & \iff 0  < \frac{\thetax + 2^\ell}{2^\ell} < 1
    \iff \bigl\lfloor \frac{\thetax +
      2^\ell}{2^\ell} \bigr\rfloor = 0    
  \end{array}
  \right.\thinspace,
\]
and so,
\iflong
\begin{align*}
  \pred{\thetax \ge 0}
  &= \bigl\lfloor \tfrac{\thetax + 2^{\ell}}{2^\ell}
    \bigr\rfloor
    = \bigl\lfloor \tfrac{\tstar - \mu}{2^\ell} \bigr\rfloor + 1
  & \text{since $\tstar = \thetax + \mu$}\\
  &= \bigl\lfloor \tfrac{\tstar}{2^\ell} \bigr\rfloor
    - \bigl\lfloor \tfrac{\mu}{2^\ell} \bigr\rfloor
    + \bigl\lfloor \tfrac{\eta - \ul{\mu}}{2^\ell}\bigr\rfloor + 1
  & \text{by \Cref{lem:lemma1}}\\
  &= \bigl\lfloor \tfrac{\tstar}{2^\ell} \bigr\rfloor
    - \bigl\lfloor \tfrac{\mu}{2^\ell} \bigr\rfloor
      + \delta
  & \text{by \Cref{lem:lemma2}}\\
  &= \bigl(\bigl\lfloor \tfrac{\tstar}{2^\ell} \bigr\rfloor
    - \bigl\lfloor \tfrac{\mu}{2^\ell} \bigr\rfloor
    + \delta\bigr) \bmod 2
  & \text{since $\pred{\thetax \ge 0} \in \{0,1\}$}\\
  &= \bigl(\bigl\lfloor \tfrac{\mu}{2^\ell} \bigr\rfloor +
  \deltaC\bigr) \bmod 2 & \text{since $\deltaS = \lfloor t^*/2^\ell
                          \rfloor \bmod{2}$}\\
  &= \mu_\ell \oplus \deltaC\enspace.
\end{align*}
\fi
Now, noting
$\sign(\thetax) = 
(-1)^{\neg\pred{\vec{\theta}^\tran\vec{x} \ge 0}}$, we get the desired
result.

\paragraph*{Security}
The security of the protocol of \cref{fig:svm_classification_core} follows from the
fact that the inner product $\thetax$ is statistically masked by the
random value $\mu$. Security parameter $\secpar$ guarantees that
the probability of an information leak due to a carry is
negligible. The security also depends on the security of the DGK+
comparison protocol, which is provably secure (cf. \cref{rem:DGKsec}).

\subsubsection{A Heuristic Protocol}
The previous protocol, thanks to the use of the DGK+ algorithm offers
provable security guarantees but incurs the exchange of
${2(\ell+1)}$ ciphertexts. Here we aim to reduce the number of
ciphertexts and introduce a new heuristic protocol.
This protocol requires the introduction of a
signed factor~$\lambda$, such that $\abs{\lambda} > \abs{\mu}$, and we  
now use both $\mu$ and $\lambda$ to mask the model. To ensure that
${\lambda \thetax + \mu}$ remains within the message space, we pick $\lambda$
in $\B$ where
\iflong
\[ 
  \B \coloneqq \left[-\left\lceil\frac{\lceil M/2
      \rceil}{B+1}\right\rceil, \left\lfloor\frac{\lceil M/2
      \rceil}{B+1}\right\rfloor \right]\enspace.
\]
\fi
Furthermore, to ensure the effectiveness of the masking, $\B$ should
be sufficiently large; namely, $\#\B > 2^\secpar$ for a security
parameter $\secpar$, hence ${M > 2^\ell(2^\kappa - 1)}$.

The protocol, which is illustrated in \cref{fig:svm_classification_heuristic}, runs as
follows:
\begin{enumerate}
\item The client encrypts its input data $\vec{x}$ using its public
  key, and sends its key and the encrypted data to the server.
\item The server draws at random a signed scaling factor
  $\lambda \in \B$, $\lambda \neq 0$, and an offset factor
  $\mu \in \B$ such that $\abs{\mu} < \abs{\lambda}$ and
  $\sign(\mu) = \sign(\lambda)$. The server then defines the bit
  $\deltaS$ such that $\sign(\lambda) = (-1)^\deltaS$ and computes an
  encryption $\ct{t}^*$ of the shifted and scaled inner product
  $\tstar = (-1)^\deltaS \cdot (\lambda \vec{\theta}^\tran \vec{x} +
  \mu)$ \iflong as
  \[
    \ct{t}^* = \HEnc[\big]{(-1)^\deltaS\,\mu} \boxplus
    \boxsum_{i=0}^d \bigl((-1)^\deltaS\lambda \,
    \theta_i\bigr) \odot \HEnc{x_i}\thinspace,
  \]
  \fi 
  and sends $\ct{t}^*$ to the client.\iflong\footnote{Note that instead, one
    could define $\lambda, \mu \getsr \B$ with $\lambda > 0$ and
    $\abs{\mu} < \lambda$, and $t^* = \lambda\thetax + \mu$. We
    however prefer the other formulation as it easily generalises to
    extended settings (see \cref{subsec:sign_activation}).}\fi

\item In the final step, the client decrypts $\ct{t}^*$ using
  its private key, recovers $\tstar$ as a signed integer of $\M$, and
  deduces the class of the input data as $\yhat = \sign(t^*)$.
\end{enumerate}

\begin{figure}[h]
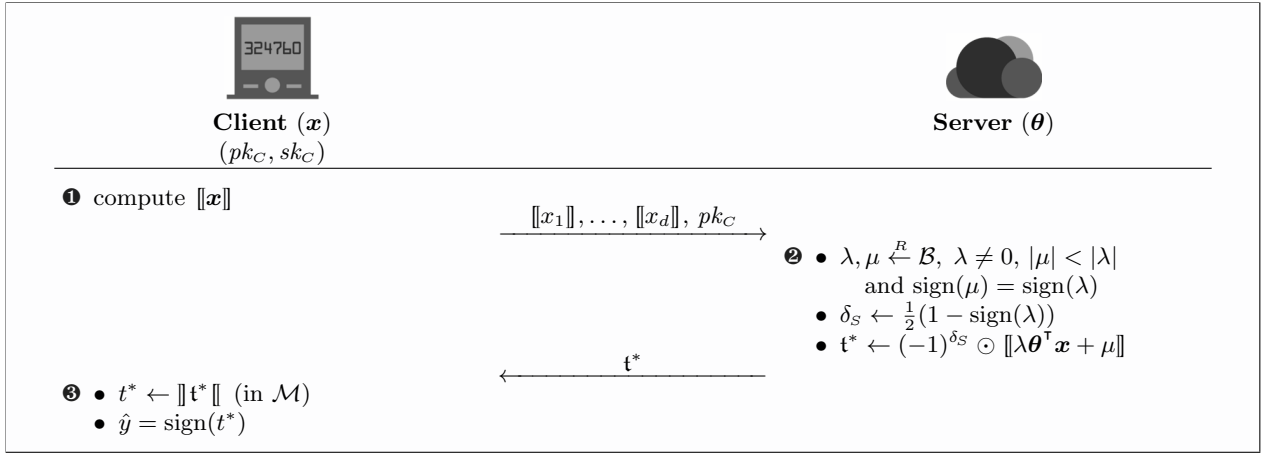

  \begin{protocol}
    \multicolumn{1}{c}{\textbf{Client} ($\vec{x}$)}
    && \multicolumn{1}{c}{\textbf{Server} ($\vec{\theta}$)}\\
    \multicolumn{1}{c}{$(\pkC, \skC)$}
    && \multicolumn{1}{c}{}\\
    \hline
    \noalign{\medskip}
    \cding{1}\enspace compute $\HEnc{\vec{x}}$\\[-1ex]
    & $\rarrow[\HEnc{x_1}, \dots, \HEnc{x_d},\, \pkC]$\\[-1ex]
    && \cding{2}\enspace
    $\begin{array}[t]{@{}l@{}}
       \bullet\enspace \lambda, \mu \getsr \B,\;
       \mathrlap{\lambda\neq0,\,\abs{\mu}<\abs{\lambda}}\\
       \qquad \text{and }\mathrlap{\sign(\mu) = \sign(\lambda)}\\
       \bullet\enspace \deltaS \gets \frac12(1-\sign(\lambda))\\
       \bullet\enspace  \ct{t}^* \gets \mathrlap{(-1)^\deltaS \odot
       \HEnc{\lambda\vec{\theta}^\tran \vec{x} + \mu}}
    \end{array}$\\[-1ex]
    & $\larrow[\ct{t}^*]$\\[-1ex]
    \cding{3}\enspace
    $\begin{array}[t]{@{}l@{}}
    \bullet\enspace t^* \gets \HDec{\ct{t}^*} \enspace\text{(in $\M$)}\\
    \bullet\enspace \yhat = \sign(t^*)
    \end{array}$
  \end{protocol}
  \caption{Heuristic protocol for privacy-preserving
    SVM classification.}\label{fig:svm_classification_heuristic}
\end{figure}

\paragraph*{Correctness}
The constraint $\abs{\mu} < \abs{\lambda}$ with $\lambda \neq 0$
ensures that $\hat{y} = \sign(\thetax)$. Indeed, as
$(-1)^\deltaS = \sign(\lambda) = \sign(\mu)$, we have
$t^* = (-1)^\deltaS(\lambda\thetax + \mu) = \abs{\lambda}\thetax +
\abs{\mu} = \abs{\lambda}(\thetax + \epsilon)$ with
$\epsilon \coloneqq \abs{\mu}/\abs{\lambda}$. Hence, whenever
$\thetax \neq 0$, we get
$\hat{y} = \sign(t^*) = \sign(\thetax + \epsilon) = \sign(\thetax)$
since $\abs{\thetax} \ge 1$ and
$\abs{\epsilon} = \abs{\mu}/\abs{\lambda} < 1$. If $\thetax = 0$ then
$\hat{y} = \sign(\epsilon) = 1$.

\paragraph*{Security}
We stress that the private comparison protocol we use in
\cref{fig:svm_classification_heuristic} does not come with formal security guarantees.
In particular, the client learns the value of
$t^* = \lambda\vec{\theta}^\tran\vec{x} + \mu$ with
$\lambda, \mu \in \B$ and $\abs{\mu} < \abs{\lambda}$. Some
information on $t \coloneqq \vec{\theta}^\tran\vec{x}$ may be leaking
from $t^*$ and, in turn, on $\vec{\theta}$ since $\vec{x}$ is known to
the client. The reason resides in the constraint
$\abs{\mu} < \abs{\lambda}$. So, from
$t^* = \lambda\vec{\theta}^\tran\vec{x} + \mu$, we deduce
$\log \abs{t^*} \le \log\abs{\lambda} + \log{(\abs{t} + 1)}$. For
example, when $t$ has two possible very different ``types'' of values
(say, very large and very small), the quantity $\log \abs{t^*}$ can
be enough to discriminate with non-negligible probability the type
of $t$. This may possibly leak information on $\vec{\theta}$. That
does not mean that the protocol is necessarily insecure but it should
be used with care.

\iflong
\begin{remark}
  The bandwidth usage could be even reduced to one ciphertext and a
  single bit with the dual approach. From the published encrypted
  model $\HHEnc{\vec{\theta}}$, the client could homomorphically
  compute and send to the server
  $\ct{t}^* = \HHEnc{\lambda\vec{\theta}^\tran\vec{x} + \mu}$ for
  random $\lambda, \mu \in \B$ with $\abs{\mu} < \abs{\lambda}$. The
  server would then decrypt $\ct{t}^*$, obtain~$t^*$, compute
  $\deltaS = \frac12(1-\sign(t^*))$, and return $\deltaS$ to the
  client. Analogously to the primal approach, the output class
  $\yhat = \sign(\vec{\theta}^\tran\vec{x})$ is obtained by the client
  as $\yhat = (-1)^\deltaS\cdot\sign(\lambda)$. However, and
  contrarily to the primal approach, the potential information leakage
  resulting from $t^*$---in this case on $\vec{x}$---is now on the
  server's side, which is in contradiction with our
  Requirement~\ref{itm:data_conf} (input confidentiality). We do not
  further discuss this variant.
\end{remark}
\fi

\section{Application to Neural Networks}\label{sec:application}

Typical \emph{feed-forward neural networks} are represented as large
graphs. Each node on the graph is often called a \emph{unit}, and
these units are organised into layers. At the very bottom is the input
layer with a unit for each of the coordinates $x_j^{(0)}$ of the input
vector $\vec{x}^{(0)} \coloneqq \vec{x} \in \X$. Then various
computations are done in a bottom-to-top pass and the output
$\yhat \in \Y$ comes out all the way at the very top of the graph.
Between the input and output layers, a number of \emph{hidden layers}
are evaluated. We index the layers with a superscript $(l)$, where
$l=0$ for the input layer and $1\leq l < L$ for the hidden
layers. Layer $L$ corresponds to the output. Each unit of each layer
has directed connections to the units of the layer below; see
\cref{fig:ffnn}.

\iflong
\begin{figure}[h]
  \begin{minipage}{\textwidth}
    \centering
    \begin{tikzpicture}[scale=1, transform shape,
  circ/.style={draw, circle,
    fill=antiquewhite,
    minimum size=5mm}
  ]
  \matrix[column sep=5mm]
  {
    \node[right] (l) {Hidden layer $l$};
    && \node{$\cdots$};
    & \node[circ,label=above:{$x_{j-1}^{\mathrlap{(l)}}$}] {};
    & \node[circ,label=above:{$x_{j}^{\mathrlap{(l)}}$},fill=lemonchiffon] (out) {};
    & \node[circ,label=above:{$x_{j+1}^{\mathrlap{(l)}}$}] {};    
    & \node{$\cdots$};\\[1.2cm]
    \node[right] (l-1) {Hidden layer $l - 1$};
    & \node[circ,label=below:{$x_{1}^{\mathrlap{(l-1)}}$}] (x1) {};
    & \node[circ,label=below:{$x_{2}^{\mathrlap{(l-1)}}$}] (x2) {};
    && \node{$\dots$};&
    & \node[circ,label=below:{$x_{d_{l-1}}^{\mathrlap{(l-1)}}$}] (xd) {};\\
  };
  \draw[-latex] (x1) -- (out) node[midway, above left, sloped]{$\theta_{j,1}^{(l)}$};
  \draw[-latex] (x2) -- (out) node[midway, below, sloped]{$\theta_{j,2}^{(l)}$};
  \draw[-latex] (xd) -- (out) node[midway, above right, sloped]{$\theta_{j,d_{l-1}}^{(l)}$};
  \draw[very thick, shorten <= .5ex, shorten >= .5ex, ->] (l-1) -- (l-1 |- l.south);
\end{tikzpicture}
    \subcaption{Going from layer $l-1$ to layer $l$.}\label{fig:ffnn}
  \end{minipage}\par\medskip
  
  \begin{minipage}{\textwidth}
    \centering
    \begin{tikzpicture}[scale=1, transform shape,
      circ/.style={draw, circle,
        inner sep=2pt,
        font=\Huge},
      squa/.style={
        draw,
        inner sep=4pt,
        font=\large,
        join = by -latex},
      start chain=bl,node distance=10mm
      ]
      \node[on chain=bl,right=7.5mm] (xbl) {$\vdots$};
      \node[on chain=bl] {\hphantom{$\theta_i^{(l)}$}};
      \node[on chain=bl,circ,fill=pearl] (sigma) {$\mathrm{\Sigma}$};
      \node[on chain=bl,squa,fill=antiquewhite,label=above:{\parbox{2cm}{\centering Activation\\[-2pt] function}}]{$g_j^{(l)}$};
      \node[on chain=bl,label=above:Output,join=by -latex]{$x_j^{(l)}$};
      \begin{scope}[start chain=1]
        \node[on chain=1] at (0,1cm) (x1) {$x_1^{(l-1)}$};
        \node[on chain=1,label=above:Weights,join=by o-latex] (w1) {$\theta_{j,1}^{(l)}$};
      \end{scope}
      \begin{scope}[start chain=2]
        \node[on chain=2] at (0,.5cm) (x2) {$x_2^{(l-1)}$};
        \node[on chain=2,join=by o-latex] (w2) {$\theta_{j,2}^{(l)}$};
      \end{scope}
      \begin{scope}[start chain=d]
        \node[on chain=d] at (0,-.75cm) (xd) {$x_{d_{l-1}}^{(l-1)}$};
        \node[on chain=d,join=by o-latex] (wd) {$\theta_{j,d_{l-1}}^{(l)}$};
      \end{scope}
      \node[above=2.5mm] at (sigma|-w1) (b) {\parbox{2cm}{\centering Bias\\[2pt]$\theta_{j,0}^{(l)}$}}; 
      \draw[-latex] (w1) -- (sigma);
      \draw[-latex] (w2) -- (sigma);
      \draw[-latex] (wd) -- (sigma);
      \draw[o-latex] (b) -- (sigma);
      \draw[decorate,decoration={brace,mirror}] (x1.north west) --
      node[left=10pt,rotate=90,anchor=center] {Inputs} (xd.south west);
    \end{tikzpicture}
    \subcaption{Zoom on computing unit $j$ in layer $l$.}\label{fig:neuron}
  \end{minipage}
  \caption{Relationship between a hidden unit in layer $l$ and the
    hidden units of layer~$l-1$ in a feed-forward neural
    network.}
\end{figure}

\Cref{fig:neuron} details the outcome $x_j^{(l)}$ of the
$j$\textsuperscript{th} computing unit in layer $l$.
\fi
We keep the convention $x_0^{(l)} \coloneqq 1$ for all layers. If we
note $\vec{\theta_j}^{(l)}$ the vector of weight coefficients
$\theta_{j,k}^{(l)}$, $0 \leq k \leq d_{l-1}$, where $d_{l}$ is the
number of units in layer $l$, then $x_j^{(l)}$ can be expressed as:
\iflong
\begin{align}\label{eq:unitj}
  x^{(l)}_j
  &= g_j^{(l)}\Bigl(\bigl(\vec{\theta_j}^{(l)}\bigr)^{\!\tran}
    \vec{x}^{(l-1)}\Bigr)\notag\\  
  &= g_j^{(l)}\Bigl(\theta_{j,0}^{(l)} +
  \textstyle\sum_{k=1}^{d_{l-1}} \theta_{j,k}^{(l)}\,x_k^{(l-1)}\Bigr)
    \thinspace, \quad 1 \leq j \leq d_{l},\enspace 1 \leq l \leq L\enspace.
\end{align}
\fi 
\iflong
Functions $g_j^{(l)}$ are non-linear functions such as the $\sign$
function or the Rectified Linear Unit ($\ReLU$) function
\[
  \begin{array}{l@{\hspace{5em}}l}
    t \mapsto 
    \begin{cases}
      t & \text{if $t \geq 0$}\\
      0 & \text{otherwise}
    \end{cases}&
    \begin{tikzpicture}[scale=.75,baseline={(current bounding box.center)}]
      \draw[->] (-1,0) -- (1.2,0) node[below right] {$t$};
      \draw[->] (0,-.3) -- (0,1.2) node[left] {$g(t)$};
      \draw[lightsalmon,very thick] (-1,0) -- (0,0) -- (1,1);
    \end{tikzpicture}
  \end{array}
\]
\fi 
Those functions are known as \emph{activation functions}.
\iflong
  Other
examples of activation functions are defined in \cref{subsec:plogreg}.
\fi 
The weight coefficients characterise the model and are known only to
the owner of the model. Each hidden layer depends on the layer below,
and ultimately on the input data $\vec{x}^{(0)}$, known solely to the
client.

\paragraph*{Generic Solution}
A generic solution can easily be devised from~\cref{eq:unitj}: for
each inner product computation, and therefore for each unit of each
hidden layer, the server computes the encrypted inner product and the
client computes the output of the activation function in the
clear. In more detail, the evaluation of a neural network can go as follows.
\begin{enumerate}[start=0]
\item The client starts by encrypting its input data and sends it to
  the server.
\item Then, as illustrated in \cref{fig:ffnn_generic}, for each hidden
  layer $l$, $1 \le l < L$:
  \begin{enumerate}
  \item The server computes $d_l$ encrypted inner products $\ct{t}_j$
    corresponding to each unit $j$ of the layer and sends those to the
    client.
  \item The client decrypts the inner products, applies the required
    activation function $g_j^{(l)}$, re-encrypts, and sends back $d_l$
    encrypted values.
  \end{enumerate}
\item During the last round ($l=L)$, the client simply decrypts the
  $\ct{t}_j$ values and applies the corresponding activation function
  $g_j^{(L)}$ to each unit $j$ of the output layer. This is the
  required result.
\end{enumerate}

\begin{figure}[h]
  \begin{protocol}
    \multicolumn{1}{c}{\textbf{Client} ($\vec{x}^{(0)}$)}
    && \multicolumn{1}{c}{\textbf{Server}
      ($\smash{\{\vec{\theta_{j_{\mkern-2mu i}}}^{\!(i)}\}_{\substack{%
            \scriptscriptstyle 1 \le i \le L\\
            \scriptscriptstyle 0 \le j_{\mkern-1.7mu i} \le
            d_{\mkern-1mu i}}}}$)}\\
    \multicolumn{1}{c}{$(\pkC, \skC)$}
    && \multicolumn{1}{c}{}\\
    \hline
    \noalign{\medskip}
    \noalign{\vskip3ex}
    &&\tikzmarknode{enternode}{$\HEnc{\vec{x}} \gets \HEnc{\vec{x}^{(l-1)}}$}\\
    \noalign{\medskip}
    && 
       $\begin{array}[t]{@{}l@{}}
          \text{\textbf{for} $j = 1$ to $d_{l}$ \textbf{do}}\\
          \quad\ct{t}_j^{(l)} \gets
          \HEnc[\big]{\bigl(\vec{\theta_j}^{(l)}\bigr)^{\!\!\tran}\vec{x}}\\
          \text{\textbf{endfor}}
        \end{array}$\\[-3ex]
    & $\larrow[\ct{t}_1^{(l)}, \dots, \ct{t}_{d_{l}}^{(l)}]$\\[-3ex]
    $\begin{array}[b]{@{}l@{}}
       \text{\textbf{for} $j = 1$ to $d_{l}$ \textbf{do}}\\
       \quad t_j^{(l)} \gets \HDec{\ct{t}_j^{(l)}}\\
       \quad x_j^{(l)} \gets g_j^{(l)}(t_j)\\
       \text{\textbf{endfor}}
     \end{array}$\\[-4ex]
    & $\rarrow[\HEnc{x_1^{(l)}}, \dots,
    \HEnc{x_{d_{l}}^{(l)}}]$\\[-4ex]
    \noalign{\smallskip}
    && \tikzmarknode{exitnode}{\quad$l \gets l+1$}\\
    \noalign{\vskip3ex}
  \end{protocol}
  \begin{tikzpicture}[overlay,remember picture]
    \draw[dashed,shorten <= 1ex,shorten >= .5ex,-latex]
    (exitnode.south) |- ++(24ex,-3ex) |- ([yshift=3ex]enternode.north) -- (enternode);
  \end{tikzpicture}
  \caption{Generic solution for privacy-preserving evaluation of
    feed-forward neural networks. Evaluation of hidden layer
    $l$.}\label{fig:ffnn_generic}
\end{figure}
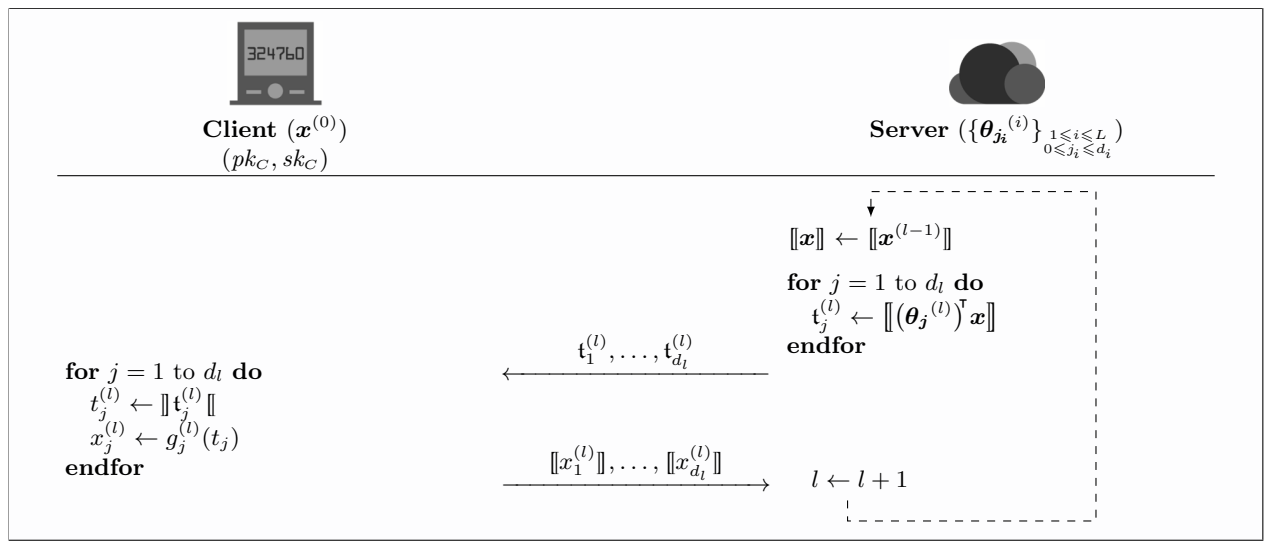

For each hidden layer $l$, exactly two messages (each comprising
$d_l$~encrypted values) are exchanged. The input and output layers
only involve one exchange; from the client to the server for the input
layer and from the server back to the client for the output layer.

\iflong
Several variations are considered in~\cite{BOP06}. For increased
security, provided that the units feature the same type of activation
functions in a given layer $l$ (i.e.,
$g_1^{(l)} = g_2^{(l)} = \dots = g_{d_l}^{(l)}$), the server may first
apply a random permutation on all units (i.e., sending the
$\ct{t}_j$'s in a random order). It then recovers the correct
ordering by applying the inverse permutation on the received
$\HEnc{x_j^{(l)}}$'s.
The server may also want to hide the activation
functions. In this case, the client holds the raw signal
$t_j \coloneqq t_j^{(l)} = (\vec{\theta_j}^{(l)})^{\!\tran}\vec{x}$
and the server the corresponding activation function $g_j^{(l)}$. The
suggestion of~\cite{BOP06} is to approximate the activation function
as a polynomial and to rely on oblivious polynomial
evaluation~\cite{NP06} for the client to get
$x_j^{(l)} \approx P_j^{(l)}(t_j)$ without learning polynomial
$P_j^{(l)}$ approximating $g_j^{(l)}$. Finally, the server may desire
not to disclose the topology of the network. To this end, the server
can distort the client's perception by adding dummy units and/or
layers.\par\medskip
\fi

In the following two sections, we improve this generic solution for two
popular activation functions: the $\sign$ and the $\ReLU$ functions.
In the new proposed implementations, everything is kept
encrypted---from start to end. The raw signals are hidden from the
client's view in all intermediate computations.

\subsection{Sign Activation}\label{subsec:sign_activation}
Binarised neural networks implement the sign function as activation
function. This is very advantageous from a hardware
perspective~\cite{HCS+16}. 

\Cref{subsec:svm_classification} describes two protocols for the client to get
the sign of $\vec{\theta}^\tran\vec{x}$. In order to use them for
binarised neural networks in a setting similar to the generic solution,
the server needs to get an encryption of $\sign(\thetax)$ for each
computing unit $j$ in layer $l$ under the client's key from
$\HEnc{\vec{x}}$, where $\HEnc{\vec{x}} \coloneqq 
\HEnc{\vec{x}^{(l-1)}}$ is the encrypted output of layer $l-1$ and 
$\vec{\theta} \coloneqq \vec{\theta_j}^{(l)}$ is the parameter vector 
for unit $j$ in layer $l$.

We start with the core protocol of~\cref{fig:svm_classification_core}. It runs in
dual mode and therefore uses the server's encryption. Exchanging the
roles of the client and the server almost gives rise to the
sought-after protocol. The sole extra change is to ensure that the
server gets the classification result encrypted. This can be achieved
by masking the value of $\deltaC$ with a random bit $b$ and sending an
encryption of $(-1)^b$. The resulting protocol is depicted in
\cref{fig:ffnn_sign}.

\iflong
\begin{figure}[p]
\else
\begin{figure}[b!]
\fi
  \begin{minipage}{\textwidth}
    \centering  
    \begin{protocol}
      \multicolumn{1}{c}{\textbf{Client}} && \multicolumn{1}{c}{\textbf{Server}
        ($\vec{\theta}, \HEnc{\vec{x}}$)}\\
      \hline
      \noalign{\medskip}
      \noalign{\vskip3ex}
      && \tikzmarknode{enternode}{
        \cding{1}\enspace$\mu
      \begin{array}[t]{@{}l}
        {}\getsr [2^\ell-1,2^{\ell+\secpar}) \iflong \\
        {}= \sum_{i=0}^{\ell+\secpar-1}\mu_i\,2^i \fi 
      \end{array}$}\\
      && \cding{2}\enspace$\ct{t}^* \gets
      \HEnc{\vec{\theta}^\tran\vec{x} + \mu}$\\ 
      && \cding{3}\enspace
      \begin{tabular}[t]{@{}l@{}}
        compute $\HHEnc{\mu_i}$\iflong, \fi
        for $0 \leq i \leq \ell - 1$
      \end{tabular}\\[-1ex]
      & $\larrow[\ct{t}^*,\, \HHEnc{\mu_0}, \dots,
      \HHEnc{\mu_{\ell-1}}]$\\[-1ex]
      \cding{4}\enspace
      \begin{tabular}[t]{@{}l}
        \textbullet\enspace  $t^* \gets
        \mathrlap{\HDec{\ct{t}^*} \bmod M}$\\ 
        \textbullet\enspace  $\eta
        \begin{array}[t]{@{}l}
          {}\gets t^* \bmod 2^\ell \iflong \\
          {}= \sum_{i=0}^{\ell-1} \eta_i\,2^i \fi 
        \end{array}$
      \end{tabular}\\
      \cding{5}\enspace
      \begin{tabular}[t]{@{}l@{}}
        \textbullet\enspace $b \getsr \{0,1\}$;\; $\mathrlap{\hat{y}^*
        \gets (-1)^b}$\\ 
        \textbullet\enspace $\mathrlap{\deltaC \gets (\lfloor
        t^*/2^\ell \rfloor \bmod 2) \oplus b}$\\ 
        \textbullet\enspace compute\footnotetext{\scriptsize The roles
        of the client and the server are interchanged, compared 
        to \cref{fig:svm_classification_core}:  so $\HEnc{\cdot}$ and
        $\HHEnc{\cdot}$ are interchanged and so 
        are $\deltaC$ and $\deltaS$; $s = 1 - 2\deltaC$.\par\smallskip
        
        \noindent In Step~7, we abuse the $\yhat$
        notation to mean either the input to the next layer or the
        final output.} 
        $\HHEnc{h_i^*}$\iflong, \fi
        for $\mathrlap{-1 \le i \le \ell-1}$
      \end{tabular}\\
      \noalign{\smallskip}
      & $\rarrow[\HEnc{\hat{y}^*},\,\substack{\textstyle\{\HHEnc{h_{-1}^*},
        \dots, \HHEnc{h_{\ell-1}^*}\}\\ 
        \text{in random order}}]$\\
      &&\cding{6}\enspace $\deltaS \gets
      \mathrlap{\pred[\bigl]{\exists i \mid h_i^* = 
          0}}$\\ 
      &&\cding{7}\enspace
      $\HEnc{\hat{y}} = \mathrlap{(-1)^{\neg(\deltaS\oplus \mu_\ell)} \odot
        \HEnc{\hat{y}^*}}$\\
      && \tikzmarknode{exitnode}{\quad$l \gets l+1$}\\
      \noalign{\vskip3ex}
    \end{protocol}
    \begin{tikzpicture}[overlay,remember picture]
      \draw[dashed,shorten <= 1ex,shorten >= .5ex,-latex]
      (exitnode.south) |- ++(33ex,-3ex) |-
      ([yshift=3ex]enternode.north) -- (enternode); 
    \end{tikzpicture}
    \iflong\subcaption{Core version.}\label{fig:ffnn_sign}\fi
  \end{minipage}
  \iflong
  \par\vskip\bigskipamount
    \begin{minipage}{\textwidth}
      \centering  
      \begin{protocol}
        \multicolumn{1}{c}{\textbf{Client}} &&
        \multicolumn{1}{c}{\textbf{Server} 
          ($\vec{\theta}, \HEnc{\vec{x}}$)}\\
        \hline
        \noalign{\medskip}
        \noalign{\vskip3ex}
        && \tikzmarknode{enternode}{
    	\cding{1}\enspace$\begin{array}[t]{@{}l@{}}
    	       \lambda, \mu \getsr \B,\;
    	       \mathrlap{\lambda\neq0,\,\abs{\mu}<\abs{\lambda}}\\
    	       \qquad \text{and }\mathrlap{\sign(\mu) = \sign(\lambda)}\\
    	    \end{array}$}\\
        && \cding{2}\enspace$\deltaS \gets \frac12(1-\sign(\lambda))$\\
        && \cding{3}\enspace$\ct{t}^* \gets
        \HEnc{\lambda \vec{\theta}^\tran\vec{x} + \mu}$\\ 
        & $\larrow[\ct{t}^*]$\\[-1ex]
        \cding{4}\enspace
        \begin{tabular}[t]{@{}l}
          \textbullet\enspace  $t^* \gets
          \mathrlap{\HDec{\ct{t}^*} \enspace\text{(in $\M$)}}$\\ 
          \textbullet\enspace  $\yhat^* = \sign(t^*)$
        \end{tabular}\\
        \noalign{\smallskip}
        & $\rarrow[\HEnc{\hat{y}^*}]$\\
        &&\cding{5}\enspace  
        $\HEnc{\hat{y}} = 
        \mathrlap{(-1)^{\deltaS} \odot \HEnc{\hat{y}^*}}$\\
        && \tikzmarknode{exitnode}{\quad$l \gets l+1$}\\
        \noalign{\vskip3ex}
      \end{protocol}
      \begin{tikzpicture}[overlay,remember picture]
        \draw[dashed,shorten <= 1ex,shorten >= .5ex,-latex]
        (exitnode.south) |- ++(28ex,-3ex) |-
        ([yshift=3ex]enternode.north) -- (enternode); 
      \end{tikzpicture}
      \subcaption{Heuristic version.}\label{fig:ffnn_sign_heuristic}
    \end{minipage}
  \fi
  \caption{Privacy-preserving $\sign$ evaluation with inputs and
    outputs encrypted under the client's public key. This serves as a
    building block for the evaluation over encrypted data of the
    $\sign$ activation function in a neural network and shows the
    computations and message exchanges for \emph{one} unit in one
    hidden layer.}
\end{figure}
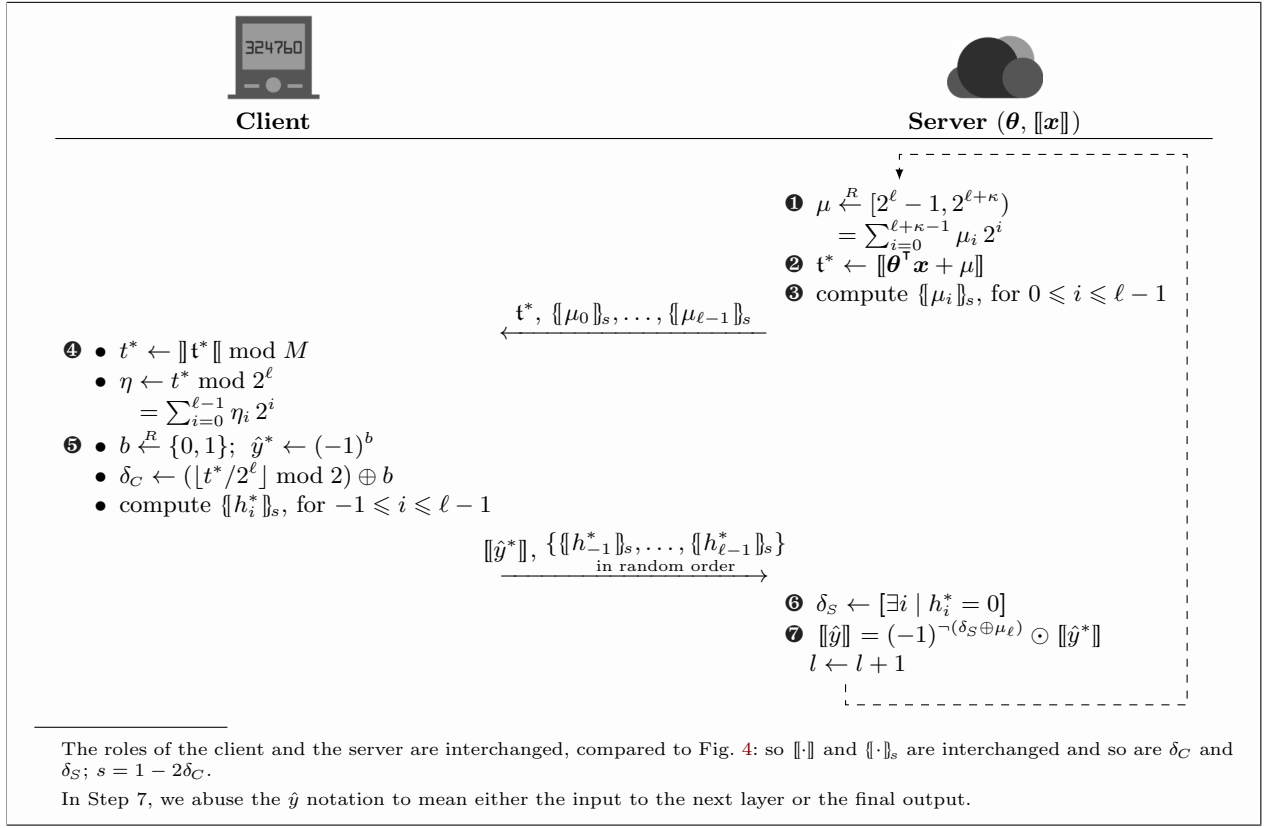
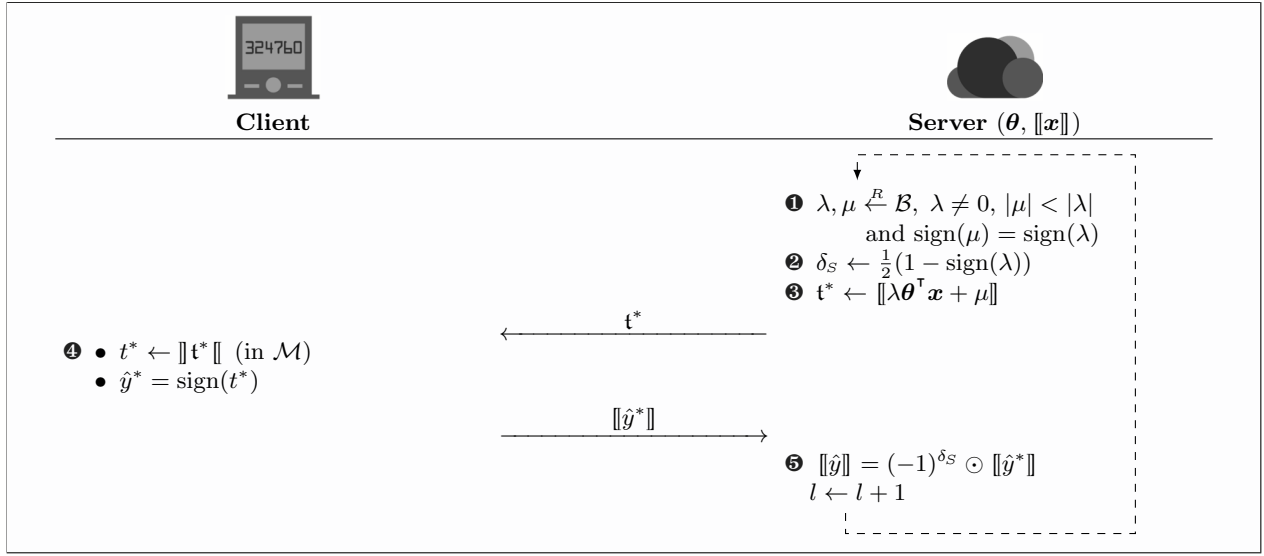

In the heuristic protocol
(cf.~\cref{fig:svm_classification_heuristic}), the server already gets
an encryption of $\HEnc{\vec{x}}$ as an input. It however fixes the
sign of $t^*$ to that of $\thetax$. If now the server flips it in a
probabilistic manner, the output class (i.e., $\sign(\thetax)$) will
be hidden from the client's view. We detail below the modifications to
be brought to the heuristic protocol to accommodate the new setting:
\begin{itemize}
\item In Step~2 of \cref{fig:svm_classification_heuristic}, the server keeps private the
  value of $\deltaS$ by replacing the definition of $\ct{t}^*$ with
  $\ct{t}^* = \HEnc{\lambda\vec{\theta}^\tran\vec{x} + \mu}$.
\item In Step~3 of \cref{fig:svm_classification_heuristic}, the client then obtains
  $\hat{y}^* \coloneqq \sign(\vec{\theta}^\tran\vec{x}) \cdot
  (-1)^\deltaS$ and returns its encryption $\HEnc{\hat{y}^*}$ to the
  server.
\item The server obtains $\HEnc{\hat{y}}$ as
  $\HEnc{\hat{y}} = (-1)^\deltaS\odot\HEnc{\hat{y}^*}$.
\end{itemize}

If $\vec{\theta} \coloneqq \vec{\theta_j}^{(l)}$ and
$\HEnc{\vec{x}} \coloneqq \HEnc{\vec{x}^{(l)}}$ then the outcome of
the protocol of \cref{fig:ffnn_sign} or of the modified heuristic
protocol is $\HEnc{\hat{y}} = \HEnc{x_j^{(l)}}$. Of course, this can
be done in parallel for all the $d_{l}$ units of layer $l$ (i.e., for
$1 \le j \le d_{l}$; see \cref{eq:unitj}), yielding
$\HEnc{\vec{x}^{(l)}} = (\HEnc{1}, \HEnc{x_1^{(l)}}, \dots,
\HEnc{x_{d_{l}}^{(l)}})$. This means that just one round of
communication between the server and the client suffices per hidden
layer.

\subsection{$\ReLU$ Activation}\label{subsec:relu_activation}
A widely used activation function is the $\ReLU$ function. It allows
a network to easily obtain sparse representations and features cheaper
computations as there is no need for computing the exponential
function~\cite{GBB11}.

The $\ReLU$ function can be expressed from the $\sign$ function as
\begin{equation}\label{eq:ReLU}
  \ReLU(t) = \tfrac12(1 + \sign(t)) \cdot t\enspace.
\end{equation}
Back to our setting, the problem is for the server to obtain
$\HEnc{\ReLU(t)}$ from~$\HEnc{t}$, where $t = \thetax$ with
$\vec{x} \coloneqq \vec{x}^{(l-1)}$ and
$\vec{\theta} \coloneqq \vec{\theta_j}^{(l)}$, in just one round of
communication per hidden layer. We saw in the previous section how to
do it for the $\sign$ function. The $\ReLU$ function is more complex
to apprehend. If we use \Cref{eq:ReLU}, the difficulty is to let the
server evaluate a \emph{product} over encrypted data. To get around that,
the server super-encrypts $\HEnc{\thetax}$, gets
$\HHEnc[\big]{\HEnc{\thetax}}$, and sends it the client. According to
its secret share the client sends back the pair
$\bigl(\HHEnc[\big]{\HEnc{0}}, \HHEnc[\big]{\HEnc{\thetax}}\bigr)$ or
$\bigl(\HHEnc[\big]{\HEnc{\thetax}}, \HHEnc[\big]{\HEnc{0}}\bigr)$.
The server then uses its secret share to select the correct item in
the received pair, decrypts it, and obtains $\HEnc{\ReLU(\thetax)}$.
For this to work, it is important that the client re-randomises
$\HHEnc[\big]{\HEnc{\thetax}}$ as otherwise the server could
distinguish it from $\HHEnc[\big]{\HEnc{0}}$.
\iflong
For an additively
homomorphic encryption algorithm $\HHEnc{\cdot}$, this can be achieved
by adding (over encrypted data) an encryption of $\HEnc{0}$,
$\HHEnc[\big]{\HEnc{\thetax}} \gets \HHEnc[\big]{\HEnc{\thetax}}
\boxplus \HHEnc[\big]{\HEnc{0}}$.
Notice that
$\HHEnc[\big]{\HEnc{\thetax}} \boxplus \HHEnc[\big]{\HEnc{0}} =
\HHEnc[\big]{\HEnc{\thetax} \boxplus \HEnc{0}} =
\HHEnc[\big]{\HEnc{\thetax}}$ where the `$\boxplus$' in the left-hand
side denotes the addition over $\HHEnc{\cdot}$ while the second one
denotes the addition over $\HEnc{\cdot}$.
\fi

Actually, a simple one-time pad suffices to implement the above
solution. To do so, the server chooses a random mask $\mu \in \M$ and
``super-encrypts'' $\HEnc{\thetax}$ as $\HEnc{\thetax + \mu}$. The
client re-randomises it as
$\ct{t}^{**} \coloneqq \HEnc{\thetax + \mu} \boxplus \HEnc{0}$, computes
$\ct{o} \coloneqq \HEnc{0}$, and returns the pair $(\ct{o},\ct{t}^{**})$
or $(\ct{t}^{**}, \ct{o})$, depending on its secret share. The server
uses its secret share to select the correct item and ``decrypts'' it.
If the server (obliviously) picked~$\ct{o}$, it already has the result in
the right form; i.e., $\HEnc{0}$. Otherwise the server has to remove
the mask $\mu$ so as to get
$\HEnc{\thetax} \gets \ct{t}^{**} \boxminus \HEnc{\mu}$. In order to
allow the server to (obliviously) remove or not the mask, the client
also sends an encryption of the pair index; e.g., $0$ for the pair
$(\ct{o},\ct{t}^{**})$ and $1$ for the pair $(\ct{t}^{**}, \ct{o})$.

\Cref{fig:ffnn_relu} details an implementation of this with the DGK+
comparison protocol. Note that to save on bandwidth the same mask
$\mu$ is used for the comparison protocol and to ``super-encrypt''
$\HEnc{\thetax}$.  The heuristic protocol can be adapted in a similar
way\iflong; see \cref{fig:ffnn_relu_heuristic}\fi

\iflong
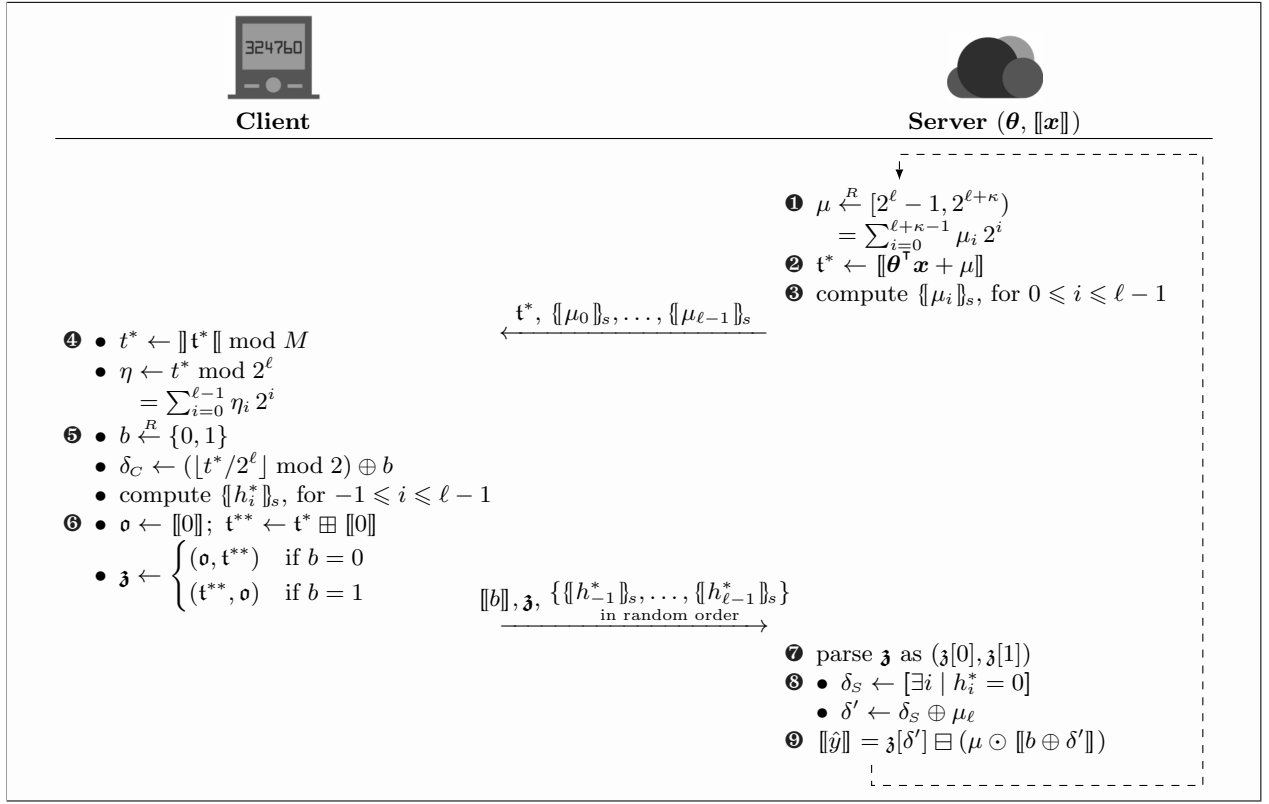
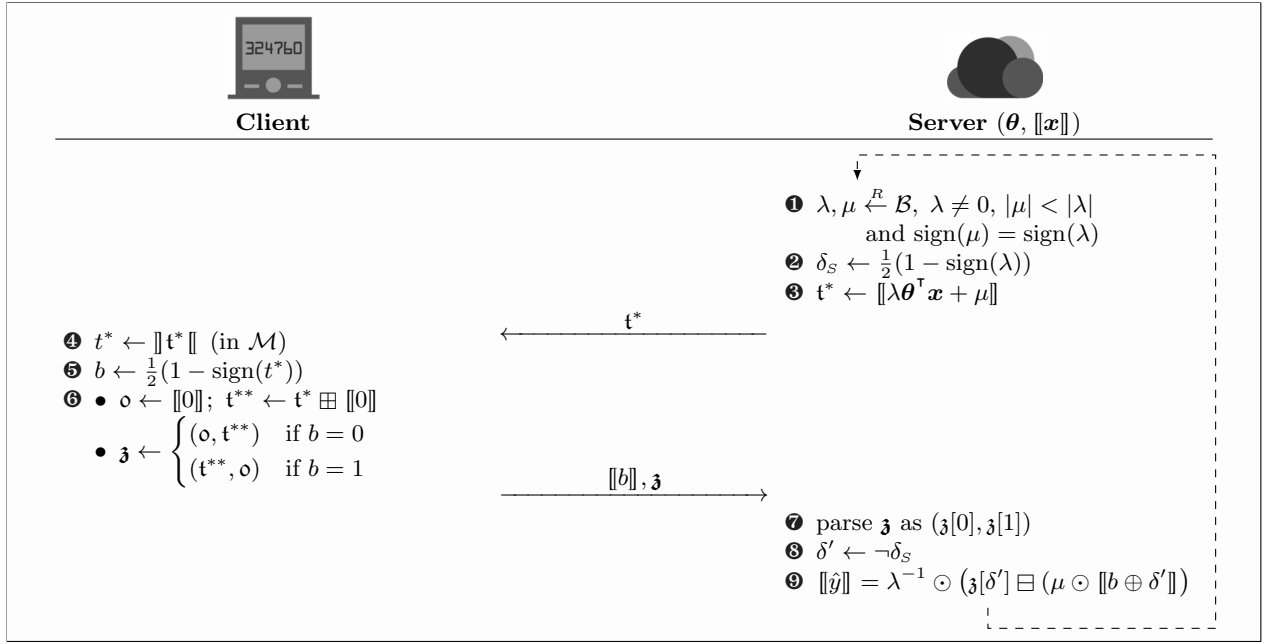
\begin{figure}
\else
\begin{figure}[!t]
\fi
  \begin{minipage}{\textwidth}
    \centering
    \begin{protocol}
      \multicolumn{1}{c}{\textbf{Client}} && \multicolumn{1}{c}{\textbf{Server}
        ($\vec{\theta}, \HEnc{\vec{x}}$)}\\
      \hline
      \noalign{\medskip}
      \noalign{\vskip3ex}
      && \tikzmarknode{enternode}{\cding{1}\enspace$\mu
        \begin{array}[t]{@{}l}
          {}\getsr [2^\ell-1,2^{\ell+\secpar})
          \iflong \\ {}= \sum_{i=0}^{\ell+\secpar-1}\mu_i\,2^i \fi 
        \end{array}$}\\
      && \cding{2}\enspace$\ct{t}^* \gets
      \HEnc{\vec{\theta}^\tran\vec{x} + \mu}$\\ 
      && \cding{3}\enspace
      \begin{tabular}[t]{@{}l@{}}
        compute $\HHEnc{\mu_i}$\iflong, \fi
        for $0 \leq i \leq \ell - 1$
      \end{tabular}\\[-1ex]
      & $\larrow[\ct{t}^*,\, \HHEnc{\mu_0}, \dots,
      \HHEnc{\mu_{\ell-1}}]$\\[-2ex]
      \cding{4}\enspace
      \begin{tabular}[t]{@{}l}
        \textbullet\enspace  $t^* \gets
        \mathrlap{\HDec{\ct{t}^*} \bmod M}$\\ 
        \textbullet\enspace  $\eta
        \begin{array}[t]{@{}l}
          {}\gets t^* \bmod 2^\ell
          \iflong \\ {}= \sum_{i=0}^{\ell-1} \eta_i\,2^i \fi 
        \end{array}$
      \end{tabular}\\
      \cding{5}\enspace
      \begin{tabular}[t]{@{}l@{}}
        \textbullet\enspace $b \getsr \{0,1\}$\\
        \textbullet\enspace $\mathrlap{\deltaC \gets (\lfloor
        t^*/2^\ell \rfloor \bmod 2) \oplus b}$\\ 
        \textbullet\enspace compute $\HHEnc{h_i^*}$\iflong, \fi
        for $\mathrlap{-1 \le i \le \ell-1}$ 
      \end{tabular}\\
      \cding{6}\enspace 
      \begin{tabular}[t]{@{}l@{}}
        \textbullet\enspace $\ct{o} \gets \HEnc{0}$;\enspace
        $\mathrlap{\ct{t}^{**} \gets \ct{t}^* \boxplus \HEnc{0}}$\\ 
        \textbullet\enspace $\vec{\ct{z}} \gets  
        \mathrlap{\begin{cases}
            (\ct{o}, \ct{t}^{**}) & \text{if $b=0$}\\
            (\ct{t}^{**}, \ct{o}) & \text{if $b=1$}
          \end{cases}}$
      \end{tabular}\\[-4ex]
      \noalign{\smallskip}
      & $\rarrow[\HEnc{b},
      \vec{\ct{z}},\,\substack{\textstyle\{\HHEnc{h_{-1}^*},  
        \dots, \HHEnc{h_{\ell-1}^*}\}\\ 
        \text{in random order}}]$\\
      &&\cding{7}\enspace parse $\vec{\ct{z}}$ as $(\ct{z}[0], \ct{z}[1])$\\
      &&\cding{8}\enspace 
      \begin{tabular}[t]{@{}l@{}}
        \textbullet\enspace$\deltaS \gets \mathrlap{\pred[\bigl]{\exists
        i \mid h_i^* = 0}}$\\
        \textbullet\enspace $\delta' \gets \deltaS \oplus \mu_\ell$
      \end{tabular}\\
      &&\tikzmarknode{exitnode}{\cding{9}\enspace $\HEnc{\hat{y}} =
        \ct{z}[\delta'] \boxminus 
        \mathrlap{(\mu \odot \HEnc{b \oplus
            \delta'})}$
      }\\
      \noalign{\vskip3ex}
    \end{protocol}
    \begin{tikzpicture}[overlay,remember picture]
      \draw[dashed,shorten <= 1ex,shorten >= .5ex,-latex]
      (exitnode.south) |- ++(32ex,-3ex) |-
      ([yshift=3ex]enternode.north) -- (enternode);
    \end{tikzpicture}
    \iflong\subcaption{Core version.}\label{fig:ffnn_relu}\fi
  \end{minipage}
  \iflong
  \par\vskip\bigskipamount
  \begin{minipage}{\textwidth}
    \centering
    \begin{protocol}
      \multicolumn{1}{c}{\textbf{Client}} && \multicolumn{1}{c}{\textbf{Server}
        ($\vec{\theta}, \HEnc{\vec{x}}$)}\\
      \hline
      \noalign{\medskip}
      \noalign{\vskip3ex}
      && \tikzmarknode{enternode}{
        \cding{1}\enspace$\begin{array}[t]{@{}l@{}}
                            \lambda, \mu \getsr \B,\;
                            \mathrlap{\lambda\neq0,\,\abs{\mu}<\abs{\lambda}}\\
                            \qquad \text{and }\mathrlap{\sign(\mu) =
                            \sign(\lambda)}\\ 
                          \end{array}$}\\
      && \cding{2}\enspace$\deltaS \gets \frac12(1 - \sign(\lambda))$\\
      && \cding{3}\enspace$\ct{t}^* \gets
      \HEnc{\lambda\vec{\theta}^\tran\vec{x} + \mu}$\\
      & $\larrow[\ct{t}^*]$\\[-2ex]
      \cding{4}\enspace $t^* \gets
      \mathrlap{\HDec{\ct{t}^*} \enspace\text{(in $\M$)}}$\\
      \cding{5}\enspace $b \gets \frac12(1 - \sign(t^*))$\\
      \cding{6}\enspace
      \begin{tabular}[t]{@{}l@{}}
        \textbullet\enspace $\ct{0} \gets \HEnc{0}$;\enspace 
        $\mathrlap{\ct{t}^{**} \gets \ct{t}^* \boxplus \HEnc{0}}$\\ 
        \textbullet\enspace $\vec{\ct{z}} \gets
        \mathrlap{\begin{cases}
            (\ct{0}, \ct{t}^{**}) & \text{if $b = 0$}\\
            (\ct{t}^{**}, \ct{0}) & \text{if $b = 1$}
          \end{cases}}$
      \end{tabular}\\[-3ex]
      \noalign{\smallskip}
      & $\rarrow[\HEnc{b}, \vec{\ct{z}}]$\\
      &&\cding{7}\enspace parse $\vec{\ct{z}}$ as $(\ct{z}[0],
      \ct{z}[1])$\\
      &&\cding{8}\enspace $\delta' \gets \neg\deltaS$\\ 
      &&\tikzmarknode{exitnode}{\cding{9}\enspace $\HEnc{\hat{y}} =
        \lambda^{-1} \odot \bigl(\ct{z}[\delta'] \boxminus
          (\mu \odot \HEnc{b \oplus \delta'} \bigr)$}\\
      \noalign{\vskip3ex}
    \end{protocol}
    \begin{tikzpicture}[overlay,remember picture]
      \draw[dashed,shorten <= 1ex,shorten >= .5ex,-latex]
      (exitnode.south) |- ++(22ex,-3ex) |-
      ([yshift=3ex]enternode.north) -- (enternode);
    \end{tikzpicture}
    \subcaption{Heuristic version.}\label{fig:ffnn_relu_heuristic}
  \end{minipage}
  \fi
  \caption{Privacy-preserving $\ReLU$ evaluation with inputs and
    outputs encrypted under the client's public key.%
  }
\end{figure}

\begin{remark}
  It is interesting to note that the new protocols readily extend to
  any piece-wise linear function, such as the clip function
  $\operatorname{clip}(t) = \max(0,\min(1,\frac{t+1}2))$
  (a.k.a. hard-sigmoid function).
  \iflong Indeed, as shown in~\cite{ABMM16}, any piece-wise linear
  function $\bbbr \rightarrow \bbbr$ with $p$ pieces can be
  represented as a sum of $p$ $\ReLU$ functions.\fi
\end{remark}

\iflong
\section{Numerical Experiments}\label{sec:evaluation}

To show the feasibility of our protocols, we consider their
implementation using Paillier's cryptosystem. In this
section we first recall this cryptosystem and then give timing
measurements of code execution and message size estimation showing the
feasibility of the proposed methods.

\subsection{Paillier's Cryptosystem}

Paillier's cryptosystem~\cite{Pai99} is an asymmetric algorithm which
is homomorphic to addition: with the encrypted values of two messages
$m_1$ and $m_2$, it is possible to compute an encrypted value of
$m_1 + m_2$. The scheme is known to be semantically secure under the
decisional composite residuosity assumption (DCRA).

\begin{description}
\item[Set-up]--- On input, given a security parameter, each party can
  create a key pair by picking at random two large primes $p$ and $q$
  and computing the product $N = p \, q$.  The public key is simply
  $\pk := N$ while the private key is $\sk := \{p,q\}$.  The message
  space is $\M = \bbbz/N\bbbz$.

\item[Encryption]--- To encrypt a message $m \in \M$, using the public
  key $\pk$, one first picks a random integer $r \getsr [1, N)$ and
  then computes the ciphertext
  \[
    \ct{m} \coloneqq \HEnc{m} = (1 + m N) \,r^N \bmod N^2\enspace.
  \]
  
\item[Decryption]--- To decrypt the ciphertext $\ct{m}$, the recipient
  first needs to recover $r$ from $\ct{m}$ using the matching secret
  key $\sk$ as
  \[
    r = \ct{m}^{N^{-1} \bmod{(p-1)(q-1)}} \bmod{N}\quad\text{with
      $\lambda = \lcm(p-1,q-1)$}
  \]
  and then recover the plaintext message $m = \frac{(\ct{m} \,
    r^{-N} \bmod N^2) - 1}{N}$. 
\end{description}

\paragraph{Homomorphism} The main homomorphism characteristics of this
scheme are summarised as follows:
\[
  \begin{cases}
    \HEnc{m_1 + m_2} = \HEnc{m_1} \boxplus \HEnc{m_2}
    = \HEnc{m_1} \cdot \HEnc{m_2} \bmod N^2\thinspace;\\    
    \HEnc{m_1 - m_2} = \HEnc{m_1} \boxminus \HEnc{m_2}
    = \HEnc{m_1} / \HEnc{m_2} \bmod N^2\thinspace;\\
    \HEnc{a \cdot m} = a \boxdot \HEnc{m}
    = \HEnc{m}^a \bmod N^2\thickspace.
  \end{cases}
\]

\subsection{Graphs}

We implemented the protocols presented in the previous sections using
the Python (version 3.7.4) programming language and the GNU
multiprecision arithmeic library (GMP version 6.1.2) on a $64$-bit
machine equipped with an Intel i7-4770 processor running at
3.4GHz. The GMP library is essentially used for generating the prime
numbers required for the keys and for performing modular
exponentiation of large integers. We used a bit precision of $P = 53$
(see~\cref{subsubsec:real_numbers}), which corresponds to the number
of significant bits for typical IEEE-754 floating point numbers
supported by Python.

We tested the protocols using randomly generated models and also
models based on the Enron-spam data set~\cite{MAP06}, a standardized
audiology data set~\cite{JPQ87}, a credit approval data
set~\cite{CA87}, a dataset of human activity recognition using
smartphones~\cite{AGOPR12},
and the breast cancer database from University of Wisconsin Hospitals,
Madison~\cite{WM90}. Performance measurements for various key sizes
are presented in
\cref{fig:pc-be-mainapr1_linear_regression_core_protocol},
\cref{fig:pc-be-mainapr1_linear_regression_dual_protocol},
\cref{fig:pc-be-mainapr1_svm_core_protocol},
\cref{fig:pc-be-mainapr1_svm_heuristic_protocol}, and
\cref{fig:ffnn_graph}. Computing times are average over 100 iterations
of each protocol on each model. The computing time depends mostly on
modular exponentiation of large integers and is linear in the size of
the model.

We selected the different key sizes to adequate protection until years
from 2020 to 2050 using Lenstra's method~\cite{AL06}. Those key sizes
correspond to a security parameter $\kappa$ between $82$ and $102$
which means that we compare between $193$ and $221$ bits (depending on
the number of features) when using the DGK+ protocol.

\newcommand{\perfgraph}[4]{
	\begin{subfigure}{.47\textwidth}
	\centering
	\includegraphics[width=\textwidth]{Code/graphs/#1/#2_#3.pdf}
	\caption{#4 dataset.}
	\label{fig:#1_#2_#3}
	\end{subfigure}}

\newcommand{\pergraphset}[3]{
\begin{figure}
\centering
	\perfgraph{#1}{#2}{random}{Random} 
	\hfill
	\perfgraph{#1}{#2}{wdbc}{Breast cancer} 
	\par\bigskip
	\perfgraph{#1}{#2}{credit}{Credit} 
	\hfill
	\perfgraph{#1}{#2}{audiology}{Audiology} 
	\par\bigskip
	\perfgraph{#1}{#2}{har}{Human activity recognition} 
	\hfill
	\perfgraph{#1}{#2}{enron}{Enron-Spam} 
\caption{\label{fig:#1_#2}Average computing time in milli-seconds for client and server side using data sets of various sizes, using the #3.}
\end{figure}
}

\pergraphset{pc-be-mainapr1}{linear_regression_core_protocol}{private linear regression core protocol}
\pergraphset{pc-be-mainapr1}{linear_regression_dual_protocol}{linear regression dual protocol}
\pergraphset{pc-be-mainapr1}{svm_core_protocol}{SVM core protocol}
\pergraphset{pc-be-mainapr1}{svm_heuristic_protocol}{SVM heuristic protocol}

\newcommand{\perfgraphffnn}[4]{
	\begin{subfigure}{.47\textwidth}
	\centering
	\includegraphics[width=\textwidth]{Code/graphs/#1/#2_#3.pdf}
	\caption{#4}
	\label{fig:#2_#3}
	\end{subfigure}
}

\mbox{}
\begin{figure}[p]
\centering
\perfgraphffnn{pc-be-mainapr1}{ffnn_generic}{random}{Generic protocol}
\hfill
\perfgraphffnn{pc-be-mainapr1}{ffnn_relu}{random}{Private $\ReLU$ activation function}\par\bigskip
\perfgraphffnn{pc-be-mainapr1}{ffnn_sign}{random}{Private sign activation function}
\hfill
\perfgraphffnn{pc-be-mainapr1}{ffnn_sign_heuristic}{random}{Heuristic protocol with sign activation}
\caption{Performance of private feed forward neural network evaluation using different activation function and protocols.}\label{fig:ffnn_graph}
\end{figure}

\subsection{Estimation of Message Sizes}

The size of exchanged messages highly depends on the implementation and
the encoding used. We chose to give a theoretical estimate of the
size of messages exchanged during the protocols.

In order to give the reader a concrete idea of size the of the messages we provide numerical estimates. For those, we imposed a strong encryption for Paillier's
algorithm and choose $\ell_M = \log_2(M) = 2048$ bits, corresponding
security parameter $\kappa = 95$ (see~\cite{AL06}). With Paillier's
scheme the size of ciphertext is $2 \ell_M$ bits. We selected a model
with $d=30$ features. To get a bit-size estimate $\ell$ of the
upperbound on inner products, we assumed that model weights and input
data are normalised so that $\abs{x_i} \leq 1$ and
$\abs{\theta} \leq 1$ when considering their real value or equivalently $\abs{x_i} \leq 2^P$ and
$\abs{\theta} \leq 2^P$ when considering their integer representation (see~\cref{subsubsec:real_numbers}). Hence $B=(d+1)2^{2P}$
and $\ell = 2P + \lceil\log_2(d+1)\rceil = 111$.  In
the case of feed forward neural networks, we chose $L=3$ and
$d_l=d=30$ ($0\leq l \leq L$).

Message sizes and their numerical estimates are summarised
in~\Cref{tab:eval}.

\begin{table}[h!p]
  \caption{Message sizes for the various protocols presented in the paper.}
  \label{tab:eval}
  \begin{center}
  \begin{minipage}{\textwidth}\centering
  \def\myfnsymbol{5}    
  \setcounter{footnote}{\myfnsymbol}\addtocounter{footnote}{-1}
  \renewcommand{\thefootnote}{\fnsymbol{footnote}}    
  \begin{tabular}{p{12em}lrr}
    \toprule
    \multicolumn{1}{c}{\textbf{Protocol}}
    & \multicolumn{1}{c}{\textbf{Protocol step}}
    & \multicolumn{1}{c}{\textbf{Size}}
    & \multicolumn{1}{c}{(kB)}\\
    \midrule
    \multirow{2}{=}{Linear/Logistic regression\\ (core)
    --- \cref{fig:linear_regression_core}} 
    & Client sends: $\pkC$, $\bigl\{\HEnc{x_i}\bigr\}_{1 \leq i \leq d}$
    & $\ell_M +  d \cdot 2\ell_M$ & $\approx 15$\\      
    \cmidrule{2-4}
    & Server sends: $\ct{t}$ & $\approx 2\ell_M$ & $<1$\\
    \midrule
    \multirow{2}{=}{Linear/Logistic regression\\ (dual)
    --- \cref{fig:linear_regression_dual}} 
    & Server publishes:
      $\pkS$, $\bigl\{\HHEnc{\theta_i}\bigr\}_{0 \leq i \leq d}$
    & $\ell_M + (d + 1) \cdot 2\ell_M$ & $\approx 16$\\
    \cmidrule(l){2-4}
    & Client sends: $\ct{t}^*$ & $\approx 2\ell_M$ & $< 1$\\
    \cmidrule{2-4}
    & Server sends: $t^*$ & $\approx \ell_M$ & $< 1$\\
    \midrule
    \multirow{2}{=}{SVM classification\\ (core)
    --- \cref{fig:svm_classification_core}}
    & Server publishes:
      $\pkS$, $\bigl\{\HHEnc{\theta_i}\bigr\}_{1 \leq i \leq d}$
    & $\ell_M + (d+1) \cdot 2\ell_M$ & $\approx 16$\\
    \cmidrule{2-4}
    & Client sends: $\ct{t}^*$, $\bigl\{\HEnc{\mu_i}\bigr\}_{0 \leq i
      \leq \ell-1}$
    & $2\ell_M + \ell \cdot 2 \ell_M$ & $\approx 56$\\
    \cmidrule{2-4}
    & Server sends: $\bigl\{\HEnc{h_i^*}\bigr\}_{-1\leq i \leq
      \ell-1}$
    & $(\ell+1) \cdot 2\ell_M$ & $\approx 56$\\
    \midrule
    \multirow{2}{=}{SVM classification\\ (heuristic)
    --- \cref{fig:svm_classification_heuristic}}
    & Client sends: $\pkC$, $\bigl\{\HEnc{x_i}\bigr\}_{1 \leq i \leq
      d}$
    & $\ell_M + d \cdot 2 \ell_M$ & $\approx 15$\\
    \cmidrule{2-4}
    & Server sends: $\ct{t}^*$ & $2\ell_M$ & $<1$\\
    \midrule
    \multirow{2}{=}{FFNN\\ (generic)
    --- \cref{fig:ffnn_generic}}
    & Server sends\footnotemark[\myfnsymbol]: 
    & \multirow{2}{*}{$L \cdot d\cdot2\ell_M$} & $45$\\
    &\quad $\ct{t}_j^{(l)}$ && ($15$ per layer)\\
    \cmidrule{2-4}
    & Client sends\footnotemark[\myfnsymbol]: 
    & \multirow{2}{*}{$L \cdot d\cdot2\ell_M $}  & $45$\\
    & \quad $\HEnc{x_j^{(l)}}$
    && ($15$ per layer)\\
    \midrule
    \multirow{2}{=}{FFNN sign act. \\(core)
    --- \cref{fig:ffnn_sign}}
    & Server sends\footnotemark[\myfnsymbol]:
    & \multirow{2}{*}{$L\cdot d \cdot(\ell + 1)\cdot 2\ell_M$}  & $5,040$\\
    & \quad $\ct{t}^*$, $\bigl\{\HHEnc{\mu_i}\bigr\}_{0 \leq i \leq
      \ell -1}$
    && ($1,680$ per layer)\\
    \cmidrule{2-4}
    & Client sends\footnotemark[\myfnsymbol]:
    & \multirow{2}{*}{$L\cdot d \cdot(\ell + 2)\cdot 2\ell_M$}  & $5,085$\\
    & \quad $\HEnc{\yhat^*}$, $\bigl\{\HHEnc{h^*_i}\bigr\}_{-1 \leq i
      \leq \ell -1}$ 
    && ($1,695$ per layer)\\
    \midrule
    \multirow{2}{=}{FFNN sign act.\\ (heuristic)
    --- \cref{fig:ffnn_sign_heuristic}}
    & Server sends\footnotemark[\myfnsymbol]:
    & \multirow{2}{*}{$L\cdot d \cdot 2\ell_M$} & $45$\\
    & \quad $\ct{t}^*$
    && ($15$ per layer)\\
    \cmidrule{2-4}
    & Client sends\footnotemark[\myfnsymbol]:
    & \multirow{2}{*}{$L\cdot d \cdot 2\ell_M$} & $45$\\
    & \quad $\HEnc{\yhat^*}$
    && ($15$ per layer)\\
    \midrule
    \multirow{2}{=}{FFNN $\ReLU$ act.\\ (core)
    --- \cref{fig:ffnn_relu}}
    & Server sends\footnotemark[\myfnsymbol]:
    & \multirow{2}{*}{$L\cdot d \cdot(\ell + 1)\cdot 2\ell_M$}
    & $5,040$\\
    & \quad $\ct{t}^*$, $\bigl\{\HHEnc{\mu_i}\bigr\}_{0 \leq i \leq \ell -1}$ 
    && ($1,680$ per layer)\\
    \cmidrule{2-4}
    & Client sends\footnotemark[\myfnsymbol]:
    & \multirow{2}{*}{$L\cdot d \cdot(\ell + 4)\cdot 2\ell_M$}
    & $5,175$\\
    & \quad $\HEnc{b}$, $\ct{z}$, $\bigl\{\HHEnc{h^*_i}\bigr\}_{-1
      \leq i \leq \ell -1}$
    && ($1,725$ per layer)\\
    \midrule
    \multirow{2}{=}{FFNN $\ReLU$ act.\\ (heuristic)
    --- \cref{fig:ffnn_relu_heuristic}}
    & Server sends\footnotemark[\myfnsymbol]:
    & \multirow{2}{*}{$L\cdot d \cdot 2\ell_M$} & $45$\\
    & \quad $\ct{t}^*$
    && ($15$ per layer)\\
    \cmidrule{2-4}
    & Client sends\footnotemark[\myfnsymbol]:
    & \multirow{2}{*}{$L\cdot d \cdot 3 \cdot 2\ell_M$} & $135$\\
    & \quad $\HEnc{b}$, $\ct{z}$
    && ($45$ per layer)\\
    \bottomrule
  \end{tabular}
  \footnotetext{\footnotemark[\myfnsymbol] Per unit, for each layer
    $l$ ($1 \le l \le L$).}
  \end{minipage}
  \end{center}
\end{table}
\fi

\section{Conclusion}\label{sec:conclusion}

In this work, we presented several protocols for privacy-preserving
regression and classification. Those protocols only require
additively homomorphic encryption and limit interactions to a mere
request and response. They are secure against semi-honest adversaries.
They can be used as-is in generalised linear models (including
logistic regression and SVM classification) or applied to other
machine-learning algorithms. As an illustration, we showed how they
nicely adapt to binarised neural networks or to feed-forward neural
networks with the $\ReLU$ activation function.

\appendix


\iflong
\clearpage
\fi
\bibliographystyle{splncs04a}
\bibliography{paper}

\begin{thebibliography}{10}
\providecommand{\url}[1]{\texttt{#1}}
\providecommand{\urlprefix}{URL }
\providecommand{\doi}[1]{https://doi.org/#1}

\bibitem{CA87}
Credit approval (1987)

\bibitem{AML12}
{Abu-Mostafa}, Y.S., {Magdon-Ismail}, M., Lin, H.T.: Learning From Data: A
  Short Course. AMLbook.com (2012)

\bibitem{AS00}
Agrawal, R., Srikant, R.: Privacy-preserving data mining. ACM Sigmod Record
  \textbf{29}(2),  439--450 (2000)

\bibitem{AGOPR12}
Anguita, D., Ghio, A., Oneto, L., Parra, X., Reyes-Ortiz, J.L.: Human activity
  recognition on smartphones using a multiclass hardware-friendly support
  vector machine. In: \iflong International Workshop of Ambient Assisted Living
  (IWAAL 2012) \else IWAAL 2012\fi (December 2012)

\bibitem{ABMM16}
Arora, R., Basu, A., Mianjy, P., Mukherjee, A.: Understanding deep neural
  networks with rectified linear units. arXiv preprint arXiv:1611.01491  (2016)

\bibitem{BOP06}
Barni, M., Orlandi, C., Piva, A.: A privacy-preserving protocol for
  neural-network-based computation. In: \iflong 8th Workshop on Multimedia and
  Security (MM\&Sec 2006) \else MM\&Sec 2006\fi. pp. 146--151. ACM (2006)

\bibitem{BLN14}
Bos, J.W., Lauter, K., Naehrig, M.: Private predictive analysis on encrypted
  medical data. \iflong Journal of Biomedical Informatics \else J. Biomed. Inf.
  \fi  \textbf{50},  234--243 (2014)

\bibitem{BPTG15}
Bost, R., Popa, R.A., Tu, S., Goldwasser, S.: Machine learning classification
  over encrypted data. In: \iflong 22nd Annual Network and Distributed System
  Security Symposium (NDSS) \else NDSS 2015\fi. The Internet Society (2015)

\bibitem{DGK08}
Damg{\aa}rd, I., Geisler, M., Kr{\o}igaard, M.: Homomorphic encryption and
  secure comparison. \iflong International Journal of Applied Cryptography
  \else Int. J. Appl. Cryptogr. \fi  \textbf{1}(1),  22--31 (2008)

\bibitem{DGK09}
Damg{\aa}rd, I., Geisler, M., Kr{\o}igaard, M.: A correction to `efficient and
  secure comparison for on-line auctions'. \iflong International Journal of
  Applied Cryptography \else Int. J. Appl. Cryptogr. \fi  \textbf{1}(4),
  323--324 (2009)

\bibitem{DF18}
Dwork, C., Feldman, V.: Privacy-preserving prediction. In: \iflong Conference
  On Learning Theory ({COLT} 2018) \else COLT 2018\fi. \iflong Proceedings of
  Machine Learning Research \else PMLR\fi, vol.~75, pp. 1693--1702. PMLR (2018)

\bibitem{EFG+09}
Erkin, Z., Franz, M., Guajardo, J., Katzenbeisser, S., Lagendijk, I., Toft, T.:
  Privacy-preserving face recognition. In: \iflong Privacy Enhancing
  Technologies (PETS~2009) \else PETS 2009\fi. \iflong Lecture Notes in
  Computer Science \else LNCS\fi, vol.~5672, pp. 235--253. Springer (2009)

\bibitem{Gen09}
Gentry, C.: Fully homomorphic encryption using ideal lattices. In: 41st Annual
  ACM Symposium on Theory of Computing (STOC). pp. 169--178. ACM (2009)

\bibitem{GBB11}
Glorot, X., Bordes, A., Bengjio, Y.: Deep sparse rectifier neural networks. In:
  \iflong 14th International Conference on Artificial Intelligence and
  Statistics (AISTAT) \else AISTAT 2011\fi. \iflong Proceedings of Machine
  Learning Research \else PMLR\fi, vol.~15, pp. 315--323. PMLR (2011)

\bibitem{GLLM04}
Goethals, B., Laur, S., Lipmaa, H., Mielik{\"a}inen, T.: On private scalar
  product computation for privacy-preserving data mining. In: \iflong
  Information Security and Cryptology -- {ICISC} 2004 \else ICISC 2004\fi.
  \iflong Lecture Notes in Computer Science \else LNCS\fi, vol.~3506, pp.
  104--102. Springer (2004)

\bibitem{GM84}
Goldwasser, S., Micali, S.: Probabilistic encryption. \iflong Journal of
  Computer and System Sciences \else J. Comput. Syst. Sci. \fi  \textbf{28}(2),
   270--299 (1984)

\bibitem{HTF09}
Hastie, T., Tibshirani, R., Friedman, J.: The Elements of Statistical Learning.
  Springer Series in Statistics, Springer, 2nd edn. (2009)

\bibitem{HCS+16}
Hubara, I., Courbariaux, M., Soudry, D., El-Yaniv, R., Bengio, Y.: Binarized
  neural networks. In: \iflong Advances in Neural Information Processing
  Systems 29 (NIPS 2016) \else NISP 2016\fi. pp. 4107--4115 (Curran Associates,
  Inc)

\bibitem{JP19}
Joye, M., Petitcolas, F.A.P.: \textsc{Pinfer}: Privacy-preserving inference.
  In: Data Privacy Management, Cryptocurrencies and Blockchain Technology
  (DPM/CBT 2019). \iflong Lecture Notes in Computer Science \else LNCS\fi, vol.
  11737, pp. 3--21. Springer (2019)

\bibitem{JS18}
Joye, M., Salehi, F.: Private yet efficient decision tree evaluation. In:
  \iflong Data and Applications Security and Privacy XXXII (DBSec 2018) \else
  DBSEC 2018\fi. \iflong Lecture Notes in Computer Science \else LNCS\fi, vol.
  10980, pp. 243--259. Springer (2018)

\bibitem{KSW+18}
Kim, M., Song, Y., Wang, S., Xia, Y., Jiang, X.: Secure logistic regression
  based on homomorphic encryption: Design and evaluation. \iflong JMIR Medical
  Informatics \else JMIR Med. Inform. \fi  \textbf{6}(2) (2018)

\bibitem{AL06}
Lenstra, A.K.: Key lengths. In: The Handbook of Information Security. Wiley
  (2006)

\bibitem{LP00}
Lindell, Y., Pinkas, B.: Privacy preserving data mining. In: \iflong Advances
  in Cryptology -- CRYPTO 2000 \else CRYPTO 2000\fi. \iflong Lecture Notes in
  Computer Science \else LNCS\fi, vol.~1880, pp. 36--54. Springer (2000)

\bibitem{MAP06}
Metsis, V., Androutsopoulos, I., Paliouras, G.: Spam filtering with naive bayes
  -- which naive bayes? In: \iflong 3rd Conference on Email and Anti-Spam (CEAS
  2006) \else CEAS 2006\fi (2006)

\bibitem{MZ17}
Mohassel, P., Zhang, Y.: {SecureML}: A system for scalable privacy-preserving
  machine learning. In: \iflong 2017 IEEE Symposium on Security and Privacy
  \else IEEE S\&P 2017\fi. pp. 19--38. IEEE Computer Society (2017)

\bibitem{NP06}
Naor, M., Pinkas, B.: Oblivious polynomial evaluation. \iflong SIAM Journal on
  Computing \else SIAM J. Comput.\fi  \textbf{35}(5),  1254--1281 (2006)

\bibitem{Pai99}
Paillier, P.: Public-key cryptosystems based on composite degree residuosity
  classes. In: Advances in Cryptology -- {EUROCRYPT}\,'99. \iflong Lecture
  Notes in Computer Science \else LNCS\fi, vol.~1592, pp. 223--238. Springer
  (1999)

\bibitem{JPQ87}
Porter, B., Quinlan, R.: Standardized audiology database (1987)

\bibitem{TZJ+16}
Tram{\`e}r, F., Zhang, F., Juels, A., Reiter, M.K., Ristenpart, T.: Stealing
  machine learning models via prediction {API}s. In: \iflong 25th USENIX
  Security Symposium \else USENIX Security 2016\fi. pp. 601--618. USENIX
  Association (2016)

\bibitem{Veu12}
Veugen, T.: Improving the {DGK} comparison protocol. In: \iflong 2012 IEEE
  International Workshop on Information Forensics and Security (WIFS) \else
  WIFS 2012\fi. pp. 49--54. IEEE (2012)

\bibitem{WM90}
Wolberg, W.H., Mangasarian, O.: Multisurface method of pattern separation for
  medical diagnosis applied to breast cytology. Proceedings of the National
  Academy of Sciences  \textbf{87} (December 1990)

\bibitem{ZWY+17}
Zhang, J., Wang, X., Yiu, S.M., Jiang, Z.L., Li, J.: Secure dot product of
  outsourced encrypted vectors and its application to {SVM}. In: \iflong 5th
  ACM International Workshop on Security in Cloud Computing (SCC@AsiaCCS 2017)
  \else SCC@AsiaCCS 2017\fi. pp. 75--82. ACM (2017)

\end{thebibliography}

\end{document}